\documentclass[conference]{IEEEtran}
\IEEEoverridecommandlockouts
\usepackage{cite}
\usepackage{amsmath,amssymb,amsfonts}
\usepackage{algorithmic}
\usepackage{graphicx}
\usepackage{textcomp}
\usepackage{xcolor}
\usepackage{fancyhdr}
\usepackage{tablefootnote}
\usepackage{soul}

\usepackage{subfigure}
\graphicspath{{./figures/}}
\DeclareGraphicsExtensions{.pdf}
\UseRawInputEncoding

\makeatletter
\def\ps@headings{%
\def\@oddhead{\mbox{}\scriptsize\rightmark \hfil \thepage}%
\def\@evenhead{\scriptsize\thepage \hfil \leftmark\mbox{}}%
\def\@oddfoot{}%
\def\@evenfoot{}}
\makeatother
\usepackage{graphicx, amssymb, amsbsy, amsmath, amsthm}
\usepackage{multirow}
\usepackage{threeparttable}
\usepackage{cite,url}

\usepackage{booktabs}
\usepackage{color,soul}
\usepackage[font=footnotesize,labelfont=sf,textfont=sf]{caption}
\usepackage[lined,ruled,linesnumbered]{algorithm2e}
\usepackage{algorithmic}
\usepackage{bbding}
\usepackage[colorlinks,linkcolor=red,anchorcolor=blue,citecolor=blue]{hyperref}

\def\BibTeX{{\rm B\kern-.05em{\sc i\kern-.025em b}\kern-.08em
    T\kern-.1667em\lower.7ex\hbox{E}\kern-.125emX}}
\begin{document}

\newcommand{\Name}{\emph{HarmonyBatch}}
\newcommand{\CostReduce}{$82.9\%$}
\newcommand{\GPUReduce}{$37.0\%$}
\newcommand{\Instance}{ecs.e-c1m1.large}
\newcommand{\red}[1]{\textcolor{red}{#1}}
\newcommand{\blue}[1]{\textcolor{blue}{#1}}

\title{\Name{}\emph{:} Batching multi-SLO DNN Inference with Heterogeneous Serverless Functions}


\author{\IEEEauthorblockN{Jiabin Chen\IEEEauthorrefmark{2}, Fei Xu\IEEEauthorrefmark{2}\IEEEauthorrefmark{1}\thanks{$^{*}$Corresponding author: Fei Xu (\href{mailto:fxu@cs.ecnu.edu.cn}{fxu@cs.ecnu.edu.cn}). This work was supported in part by the NSFC under Grant 62372184, the Science and Technology Commission of Shanghai Municipality under Grant 22DZ2229004, the NSF under Grants OIA-2019511 and OIA-2327452, the Louisiana Board of Regents under Contract LEQSF(2019-22)-RD-A-21, the National Key Research \& Development (R\&D) Plan under Grant 2022YFB4501703, and the Major Key Project of PCL (PCL2022A05).}, Yikun Gu\IEEEauthorrefmark{2}, Li Chen\IEEEauthorrefmark{3}, Fangming Liu\IEEEauthorrefmark{4}, Zhi Zhou\IEEEauthorrefmark{5}}
\IEEEauthorblockA{\IEEEauthorrefmark{2}Shanghai Key Laboratory of Multidimensional Information Processing, East China Normal University.\\
\IEEEauthorrefmark{3}University of Louisiana at Lafayette. 
\IEEEauthorrefmark{4}Peng Cheng Laboratory.
\IEEEauthorrefmark{5}Sun Yat-sen University.\\
Email: \IEEEauthorrefmark{2}\emph{fxu@cs.ecnu.edu.cn}, \IEEEauthorrefmark{2}\emph{li.chen@louisiana.edu}, 
\IEEEauthorrefmark{4}\emph{fangminghk@gmail.com},
\IEEEauthorrefmark{5}\emph{zhouzhi9@mail.sysu.edu.cn}
}}


\maketitle

\begin{abstract}

Deep Neural Network (DNN) inference on serverless functions is gaining prominence due to its potential for substantial budget savings. Existing works on serverless DNN inference solely optimize batching requests from one application with a \emph{single} Service Level Objective (SLO) on CPU functions. However, production serverless DNN inference traces indicate that the request arrival rate of applications is \emph{surprisingly low}, which inevitably causes a long batching time and SLO violations. Hence, there is an urgent need for batching \emph{multiple} DNN inference requests with diverse SLOs (\emph{i.e.,} \emph{multi-SLO} DNN inference) in serverless platforms.
Moreover, the potential performance and cost benefits of deploying \emph{heterogeneous} (\emph{i.e.,} CPU and GPU) functions for DNN inference have received scant attention.

In this paper, we present \Name{}, a cost-efficient resource provisioning framework designed to achieve predictable performance for multi-SLO DNN inference with heterogeneous serverless functions. Specifically, we construct an analytical performance and cost model of DNN inference on both CPU and GPU functions, by explicitly considering the GPU time-slicing scheduling mechanism and request arrival rate distribution. 
Based on such a model, we devise a two-stage merging strategy in \Name{} to judiciously batch the multi-SLO DNN inference requests into application \emph{groups}. It aims to minimize the budget of function provisioning for each application group while guaranteeing diverse performance SLOs of inference applications. 
We have implemented a prototype of \Name{} on Alibaba Cloud Function Compute. Extensive prototype experiments with representative DNN inference workloads demonstrate that \Name{} can provide predictable performance to serverless DNN inference workloads while reducing the monetary cost by up to \CostReduce{} compared to the state-of-the-art methods.

\end{abstract}


\section{Introduction}
\label{sec:introduction}

The rapid evolution of artificial intelligence across diverse fields has elevated Deep Neural Network (DNN) inference to a critical role in cloud-based workloads. Fueled by the bursty nature of DNN inference requests, the fast elasticity of serverless computing makes it compelling for hosting such inference workloads, without heavy server maintenance~\cite{jarachanthan2021amps}. As DNN models grow in complexity, especially the emergence of large language models (LLMs)~\cite{brown2020language}, the computational and memory resource demands of serverless inference services increase sharply~\cite{fu2024serverlessllm}. To meet the stringent Service Level Objectives (SLOs) of DNN inference workloads, Alibaba has recently introduced GPU serverless functions~\cite{Ali_function}. Such a development in GPU functions not only provides opportunities for cost reduction~\cite{gu2023fast} but also brings new challenges in deploying DNN inference on heterogeneous serverless functions~\cite{du2022serverless}.

Upon deploying a DNN model in public clouds, the mainstream serverless DNN inference service (\emph{e.g.,} Amazon SageMaker Serverless Inference\footnote{https://aws.amazon.com/sagemaker}) cannot provide SLO guarantees for users. Users only rely on their own experience to provision CPU/GPU function resources and configure the batch size for their inference workloads. 
While several recent studies guarantee SLOs for DNN inference by either minimizing the performance interference~\cite{xu2022igniter} or optimizing function resource provisioning plan~\cite{ali2020batch}, they mainly focus on the DNN inference scenario with a \emph{single} SLO, which is impractical for the real-world situation of applications with low request arrival rates (typically less than one request per second), as evidenced in Sec.~\ref{sec:motivation-serverless}. Such low request arrival rates inevitably cause a long batching time, which cannot meet the SLO requirements (typically in milliseconds) of DNN inference workloads with large batch sizes.

\begin{figure}[!t]
  	\centerline{
  		\includegraphics[width=0.94\linewidth]{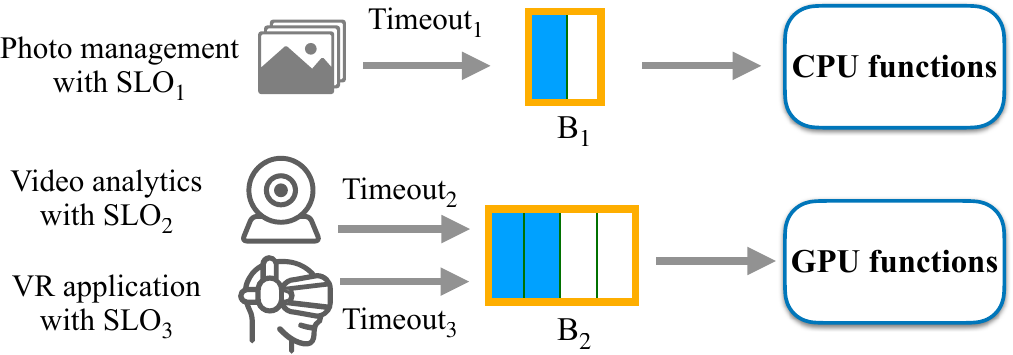}}
	\caption{A scenario of batching multiple DNN inference requests with diverse SLOs (\emph{i.e.,} multi-SLO DNN inference) on heterogeneous serverless functions. $\mathtt{Timeout_{i}}$ denotes the \emph{batching timeout} for an application $i$. $B_{1}$ and $B_{2}$ denote the batch sizes of DNN inference executed on CPU and GPU functions, respectively.}
	\label{motivation-scene}\vspace{-10pt}
\end{figure}

Fortunately, inference applications executed in public clouds can often be \emph{grouped} and each group shares the same DNN model. Such application groups bring users an opportunity to aggregate multiple DNN inference requests with diverse SLOs into large batches. As an example, Hugging Face Inference Endpoints\footnote{https://huggingface.co/inference-endpoints/serverless} provide the serverless inference API to user applications with diverse SLO requirements.
In such a scenario, we can batch the inference requests from video analytics with $\mathtt{SLO_{2}}$ and VR smart assistants~\cite{pei2023asyfunc} with $\mathtt{SLO_{3}}$ on GPU functions and thus obtain a lower monetary cost, as illustrated in Fig.~\ref{motivation-scene}. The rationale is that executing a large inference batch on GPU functions can achieve up to $\GPUReduce{}$ of cost saving as compared to a small batch on CPU functions, as evidenced by Sec.~\ref{sec:motivation-performance}.

However, existing works (\emph{e.g.,} BATCH~\cite{ali2020batch}, INFless~\cite{yang2022infless}) have solely focused on batching requests from one application with a single SLO requirement and optimizing CPU function resources~\cite{cai2023cost}. They cannot be readily applied to achieving multi-SLO serverless DNN inference, due to the two key challenges summarized below:

\noindent$\bullet$ \textbf{\emph{Complex} batching for multi-SLO DNN inference.} 
State-of-the-art inference batching techniques \emph{individually} choose a batch size and a batching timeout value for each application with a single SLO~\cite{ali2020batch, ali2022optimizing}. However, it is complex to batch requests from a set of applications in the multi-SLO DNN inference scenario. This is because it requires categorizing the applications into \emph{groups} cost-efficiently and determining the batching configuration (including the batch size for each group and a batching timeout for each application) without SLO violations. Moreover, it is difficult to estimate the monetary cost of DNN inference in such a multi-SLO scenario.

\noindent$\bullet$ \textbf{\emph{Heterogeneous} function resource provisioning.} 
Significant efforts (\emph{e.g.,} AWS Lambda Power Tuning~\cite{aws-lambda-power-tuning}, COSE~\cite{akhtar2020cose}) have been devoted to provisioning homogeneous CPU functions for guaranteeing performance SLOs of applications, which overlook the potential performance and cost benefits of deploying \emph{heterogeneous} functions. As the performance SLOs and request arrival rates of an inference application vary, the optimal provisioning plan can be \emph{shifted} between CPU and GPU functions, leading to up to $83.0\%$ of budget saving as elaborated in Sec.~\ref{sec:motivation-performance}. Moreover, the heterogeneous function provisioning can be even harder in the multi-SLO DNN inference scenario, which requires the co-optimization of the function resource allocation and application grouping as well as batching configuration.

\begin{figure}[!t]
  	\centerline{
  		\includegraphics[width=1\linewidth]{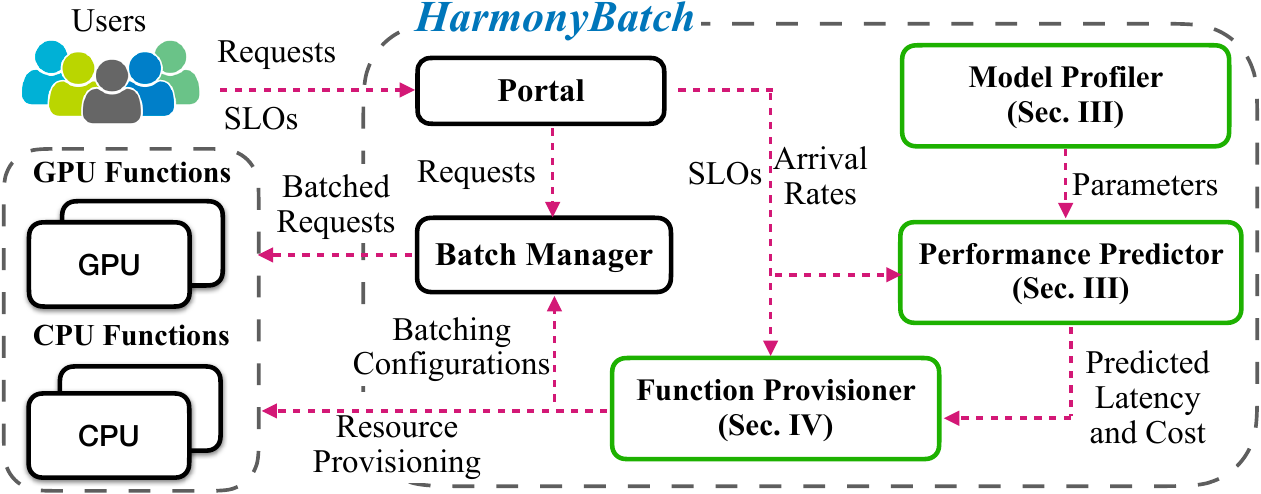}}
	\caption{Overview of \Name{}.}
	\label{introdution-prototype}\vspace{-10pt}
\end{figure}

To address these challenges above, we introduce \Name{} as illustrated in Fig.~\ref{introdution-prototype}, a cost-efficient function resource provisioning framework for achieving predictable DNN inference in serverless platforms.
To the best of our knowledge, \Name{} is the first work to provision heterogeneous (CPU and GPU) functions for judiciously batching multi-SLO DNN inference in public clouds. The main contributions of our paper are as follows:

\noindent $\rhd$ By obtaining model parameters through a lightweight workload profiling in the \emph{model profiler} (Sec.~\ref{sec:model-latency}), we develop an analytical \emph{performance predictor} (Sec.~\ref{sec:model}) for DNN inference workloads on both CPU and GPU serverless functions. We explicitly consider the GPU time-slicing schedule mechanism and request arrival rate distribution. 

\noindent $\rhd$ We design a cost-efficient inference resource provisioning strategy in the \emph{function provisioner} (Sec.~\ref{sec:algorithm}) for multi-SLO DNN inference. Our strategy uses a two-stage merging strategy to judiciously batch the multi-SLO DNN inference requests into application groups, with the aim of minimizing the budget of function provisioning for each application group while guaranteeing DNN inference performance SLOs.

\noindent $\rhd$ We implement a prototype of \Name{} (\url{https://github.com/icloud-ecnu/HarmonyBatch}) with both CPU and GPU functions on Alibaba Cloud Function Compute~\cite{Ali_function}. Extensive prototype experiments with four representative DNN models (including LLMs) demonstrate that \Name{} can achieve predictable performance for multi-SLO DNN inference, while saving the inference budget by up to \CostReduce{}, compared to state-of-the-art methods (\emph{e.g.,} BATCH~\cite{ali2020batch}, MBS~\cite{ali2022optimizing}).


\begin{figure*}
        \begin{minipage}[t]{0.32\linewidth}
		\centering
		\includegraphics[width=2.3in]{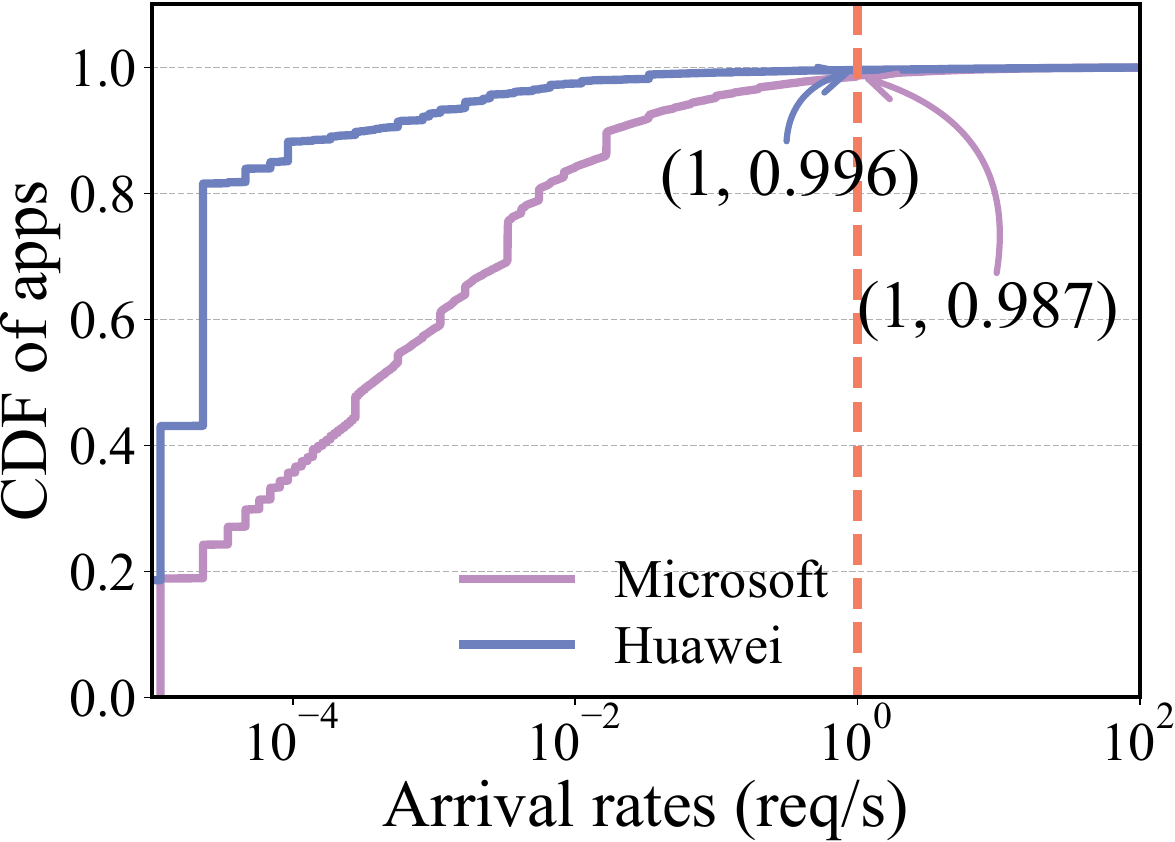}
		\caption{CDF of request arrival rates per function in Microsoft Azure Functions trace~\cite{shahrad2020serverless} and Huawei Public Cloud trace~\cite{joosen2023does}.}
		\label{workloads-cdf}
	\end{minipage}\hspace{+4pt}
	\begin{minipage}[t]{0.32\linewidth}
		\centering
		\includegraphics[width=2.3in]{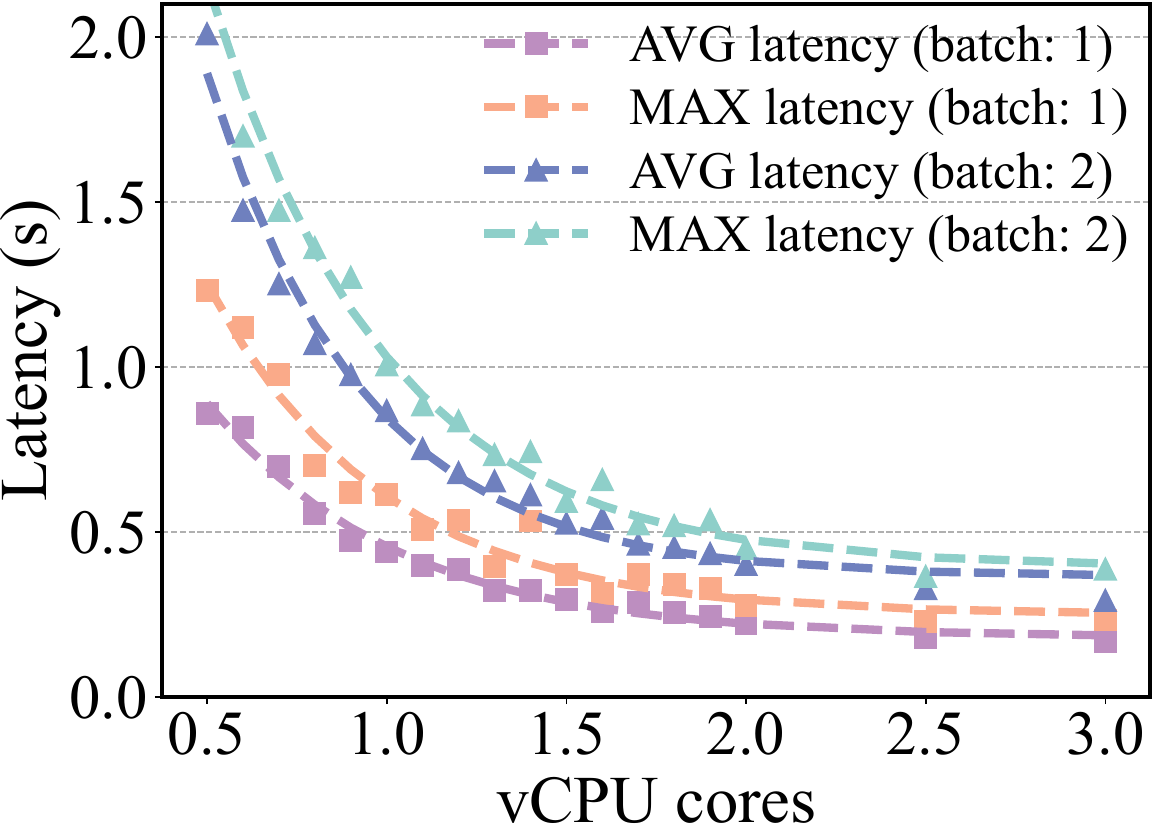}
		\caption{Inference latency of VGG-19 executed on a CPU function by varying the allocated vCPU cores from $0.5$ to $3.0$.}
		\label{motivation-cpu-vgg19}
	\end{minipage}\hspace{+4pt}
	\begin{minipage}[t]{0.32\linewidth}
		\centering
		\includegraphics[width=2.3in]{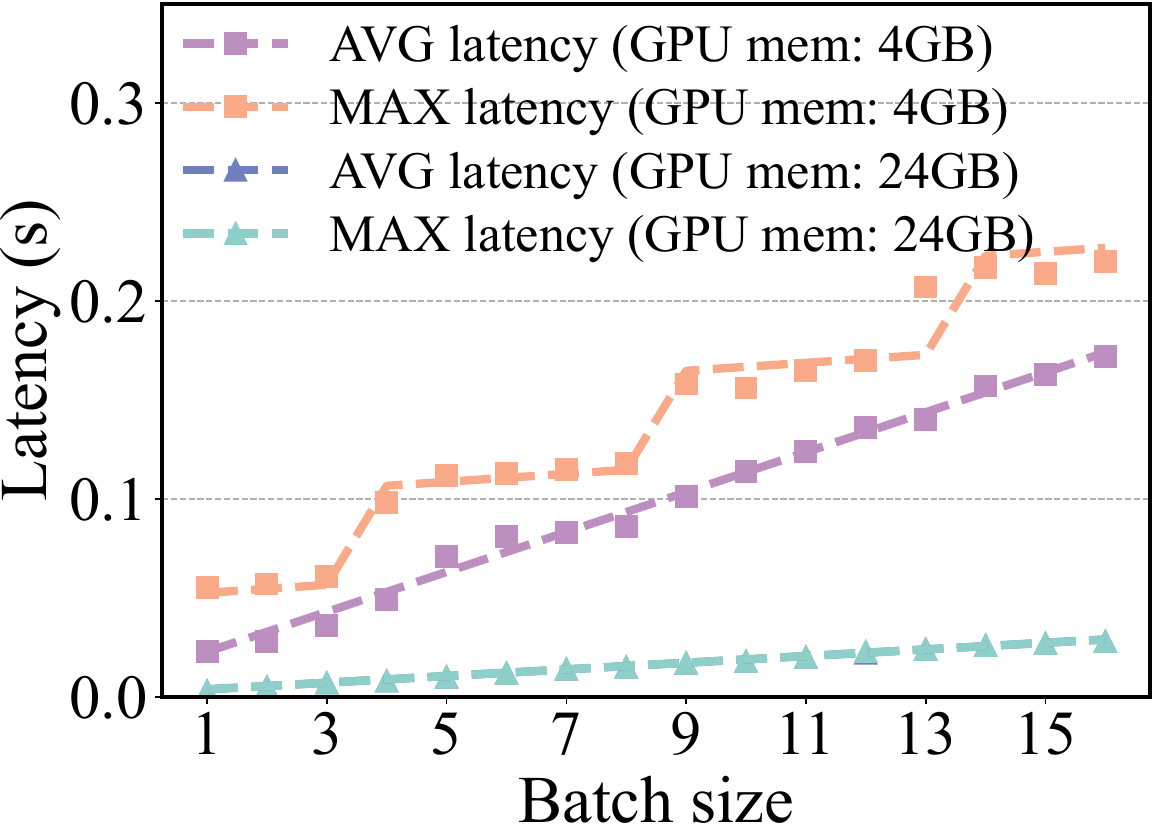}
		\caption{Inference latency of VGG-19 executed on a GPU function by configuring the batch size from $1$ to $16$.}
		\label{motivation-gpu-vgg19}
        \end{minipage}
        \vspace{-12pt}
\end{figure*}

\section{Background and Motivation}
\label{sec:motivation}

In this section, we first investigate the key factors that impact the performance and cost of serverless DNN inference. Then, we show that adequately batching requests onto \emph{heterogeneous} functions can significantly save the user budget.

\subsection{DNN Inference in Serverless Platforms}
\label{sec:motivation-serverless}

Deploying and batching DNN inference requests on serverless platforms can reduce the user budget~\cite{ali2020batch}. However, the request arrival rate of applications is typically low (\emph{i.e.,} less than $1$ request per second), as analyzed by the average request rates of applications in the Microsoft Azure Functions trace~\cite{shahrad2020serverless} and Huawei Public Cloud trace~\cite{joosen2023does}. Specifically, approximately $98.7\%$ of applications in Microsoft Azure and $99.6\%$ of applications in Huawei Cloud demonstrate a request rate of fewer than $1$ request per second, as depicted in Fig.~\ref{workloads-cdf}. As DNN inference workloads are typically latency-critical (in milliseconds), such low request arrival rates call for batching multi-SLO DNN inference requests on serverless platforms. 
In general, there exist two types of serverless functions in mainstream public clouds. Specifically, CPU functions have finer resource allocation granularity (\emph{i.e.,} $0.05$ vCPU cores on Alibaba Cloud Function Compute~\cite{Ali_function}) and lower prices, while GPU functions are equipped with higher computing power (\emph{i.e.,} $1$ total GPU for a GPU function on Alibaba Cloud Function Compute~\cite{Ali_function}) and a larger amount of GPU memory. In such a case, it is \emph{complex} for users to determine the appropriate CPU or GPU functions to be deployed together with the allocated function resources (\emph{i.e.,} vCPU cores and GPU memory). 



\subsection{Characterizing Serverless Inference Performance and Cost}
\label{sec:motivation-performance}

To explore the characteristics of serverless DNN inference performance (\emph{i.e.,} latency) and monetary cost, we conduct motivation experiments of representative DNN models (\emph{i.e.,} VGG-19~\cite{simonyan2014very}, BERT~\cite{devlin2019BERT}) on CPU and GPU functions in Alibaba Cloud Function Compute~\cite{Ali_function}. 
We execute DNN inference workloads for $100$ times with each function configuration in latency motivation experiments,
while in monetary cost motivation experiments, we execute DNN inference for $10$ minutes with each function configuration. 




\textbf{Latency.} As shown in Fig.~\ref{motivation-cpu-vgg19}, the average and maximum inference latency of VGG-19 both show a \emph{roughly exponential} decrease as more vCPU cores are allocated to CPU functions. This is because DNN inference workloads are commonly multi-core friendly, and the inference latency drops rapidly as we increase function resources at first. Our results indicate that provisioning more vCPU cores for CPU functions and increasing inference batch sizes can bring \emph{marginal} performance benefits.
As depicted in Fig.~\ref{motivation-gpu-vgg19}, the average and maximum inference latency on GPU functions \emph{overlap} each other, both exhibiting a linear relationship with the batch size with the $24$-GB configuration. Interestingly, when switching to the $4$-GB configuration, the average inference latency is \emph{roughly} linear to the batch size, while the maximum inference latency shows a \emph{stepwise} increase with the batch size. We attribute such a result above to the time-slicing scheduling mechanism of GPU functions~\cite{cGPU}, which will be elaborated in Sec.~\ref{sec:model}. 
In addition, there exists a noticeable difference between the average and maximum inference latency on both CPU and GPU functions, which indicates that DNN inference latency can be \emph{unstable} due to the performance interference among functions~\cite{xu2022igniter,li2023golgi}.

\textbf{Cost.} As illustrated in Fig.~\ref{motivation-slo}, the optimal function provisioning plans\footnote{The optimal resource provisioning plan is obtained by iterating through all possible configurations via the grid search method.} under diverse SLOs ranging from stringent to relaxed turns out to GPU, CPU, and GPU functions. Under strict SLOs, CPU functions cannot meet the SLO requirements, while GPU functions with low batch sizes can meet such requirements. As SLOs become larger (more relaxed), CPU functions gain the advantage with their lower unit price at first, and then GPU functions become cost-efficient again, as longer batching time and requests with larger batch sizes are allowed. 
As shown in Fig.~\ref{motivation-rps}, the normalized cost decreases as the request arrival rate increases especially for GPU functions. This is because the batch size for inference gets large on GPU functions as the request arrival rate increases. The results indicate that batching requests from multiple applications for larger request arrival rates can bring significant cost benefits.
In addition, the cost of DNN inference for CPU functions is insensitive to the changes in SLOs or request arrival rates.

\begin{figure*}
	\begin{minipage}[t]{0.29\linewidth}
		\centering
		\includegraphics[width=2.1in]{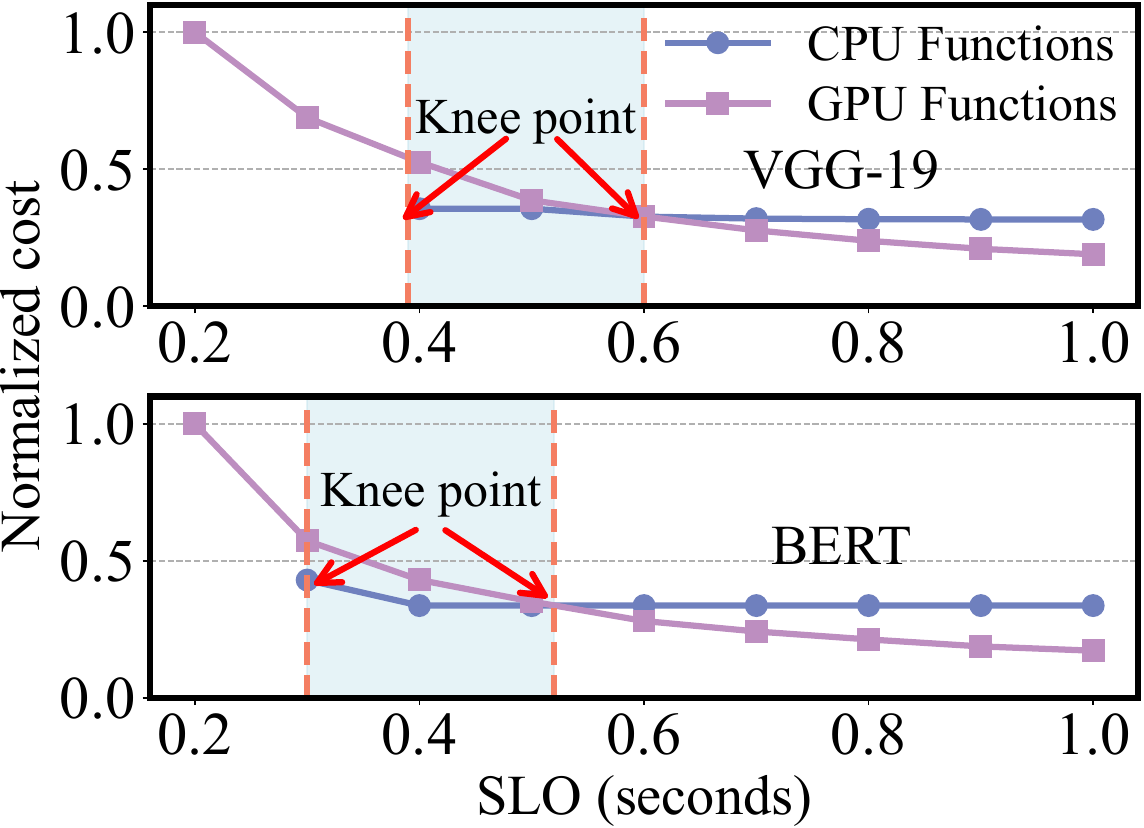}
		\caption{Normalized cost of optimal function provisioning plan under different SLOs with the request arrival rate set as $20$ req/s. The blue area denotes the optimal plan is CPU functions and two \emph{knee} points exist.}
		\label{motivation-slo}\vspace{-12pt}
	\end{minipage}\hspace{+4pt}
	\begin{minipage}[t]{0.29\linewidth}
		\centering
		\includegraphics[width=2.1in]{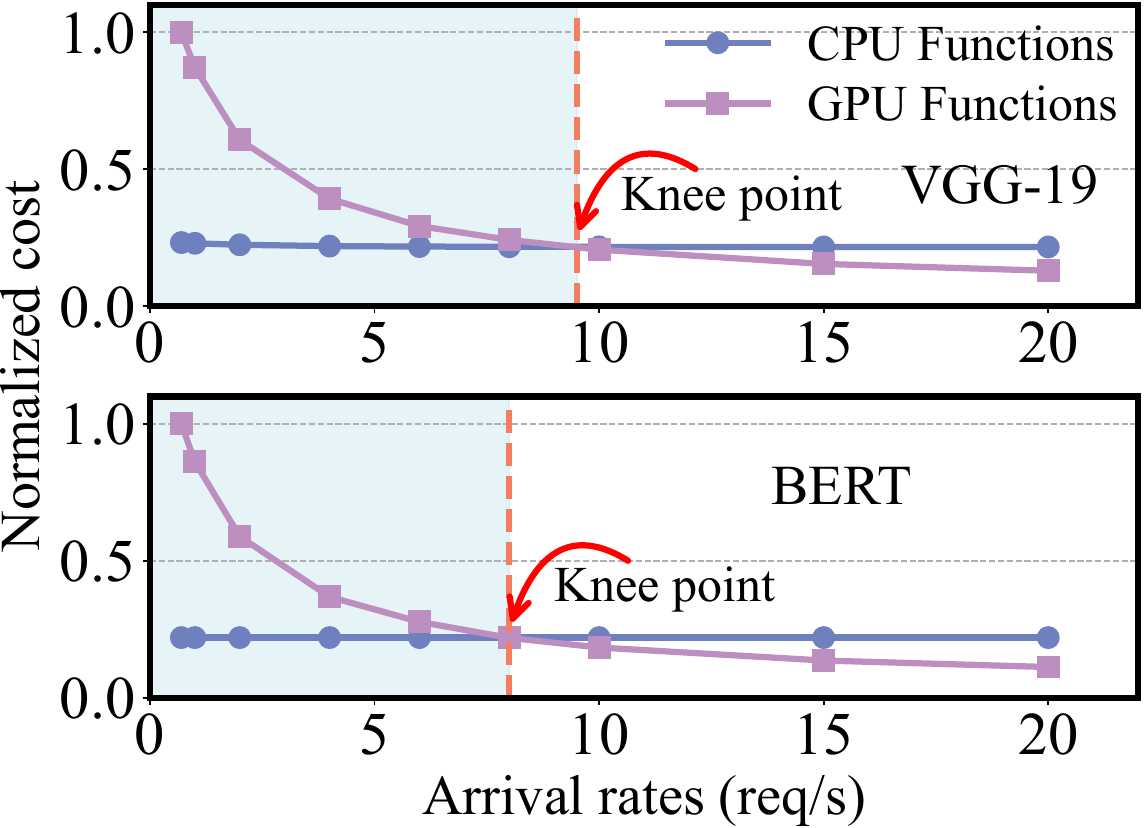}
		\caption{Normalized cost of optimal function provisioning plan under different request arrival rates with the SLO set as $1$ second. The blue area denotes the optimal plan is CPU functions and one \emph{knee} point exists.}
		\label{motivation-rps}\vspace{-12pt}
	\end{minipage}\hspace{+4pt}
	\begin{minipage}[t]{0.41\linewidth}
		\vspace{-106pt}
		\centering
		\renewcommand{\arraystretch}{1.3}
		\resizebox{2.85in}{!}{
		\begin{tabular}{c|c|c}
		\toprule[1pt]
		Strategies & Provisioning plans & Norm. cost \\
		\midrule[1pt]
		\multirow{2}{*}{BATCH}	    & $\mathtt{(1.55, 1, \left[0\right])_c, (1.5, 2, }$                  & \multirow{2}{*}{$\mathtt{1.00}$}  \\
                            		& $\mathtt{\left[0.48\right])_c, (1.5, 2, \left[0.48\right])_c}$          & \\ \hline
		\multirow{2}{*}{MBS$^+$}  	& $\mathtt{(1.6, 1, \left[0, 0, 0\right])_c}$,      	                        & \multirow{2}{*}{$\mathtt{0.88}$}  \\ 
                            		& $\mathtt{(2, 11, \left[0.65\right])_g}$                               &  \\ \hline
		\multirow{2}{*}{\Name{}}    & $\mathtt{(1.6, 1, \left[0\right])_c,}$                        	    & \multirow{2}{*}{\textcolor{red}{$\mathbf{0.63}$}} \\
                            		& $\mathtt{(2, 13, \left[0.45, 0.65\right])_g}$                    		&  \\
		\bottomrule[1pt]
		\end{tabular}
		}
		\vspace{+6pt}
		\captionof{table}{Comparison of the normalized cost of VGG-19 with BATCH~\cite{ali2020batch}, MBS$^+$~\cite{ali2022optimizing}, and \Name{} strategies. The function provisioning plan is represented as a $3$-tuple value: $\mathtt{(vCPU \  cores, batch, \left[timeouts\right])_c}$ for CPU functions and $\mathtt{(GPU \ memory, batch, \left[timeouts\right])_g}$ for GPU functions.}
		\label{tab:motivation}\vspace{-12pt}
	\end{minipage}
\end{figure*}

\subsection{An Illustrative Example}
\label{sec:motivation-example}

Batching multiple DNN inference requests from applications with diverse SLOs is challenging. In response, we design \Name{} to identify cost-effective function resource provisioning for the predictable performance of DNN inference workloads in serverless platforms. As an example, we conduct motivation experiments using three applications (\emph{i.e.,} App$_1$, App$_2$, and App$_3$) sharing the same model VGG-19~\cite{simonyan2014very}, but with different SLOs and arrival rates. Specifically, the performance SLOs of App$_1$, App$_2$, and App$_3$ are $0.5$, $0.8$, and $1$ seconds, respectively. The request arrival rates of App$_1$, App$_2$, and App$_3$ are $5$, $10$, and $20$ req/s, respectively. 

As listed in Table~\ref{tab:motivation}, \Name{} can significantly reduce the function resource provisioning cost by up to $37.0\%$ compared to the baselines. Specifically, the BATCH strategy~\cite{ali2020batch} achieves the highest cost with several SLO violations. This is because it only batches requests from one application and simply assumes the inference latency following a deterministic distribution (\emph{i.e.,} stable values). 
Though we modify the MBS$^+$ strategy to support heterogeneous function provisioning, it still brings a higher inference monetary cost compared to \Name{}. 
This is because the MBS$^+$ strategy divides the requests from the three applications evenly, which aggregates App$_1$, App$_2$ and part of App$_3$ on a CPU function with the batch size set as $1$. The remaining requests of App$_3$ are routed to a GPU function with a moderate batch size set as $11$.
In contrast, \Name{} provisions three applications with heterogeneous functions, by adequately aggregating App$_2$ and App$_3$ as a large inference batch set as $13$ on a GPU function.


\textbf{Summary.} \emph{First}, judiciously batching \emph{multi-SLO} inference requests can effectively save monetary cost while guaranteeing performance SLOs for DNN inference workloads. \emph{Second}, provisioning \emph{heterogeneous} functions for DNN inference applications can yield significant (up to \GPUReduce{}) cost benefits compared to homogeneous provisioning solutions. In particular, CPU functions are cost-effective for inference applications with \emph{moderate} SLOs and \emph{low} request arrival rates, while GPU functions are cost-effective for inference applications with \emph{tight or loose} SLOs and \emph{high} request arrival rates.

\section{System Model}
\label{sec:model}

In this section, we model the inference latency with the CPU and GPU functions and leverage the request arrival rate to model the inference cost. The key notations in our serverless DNN inference model are summarized in Table~\ref{tab:notations}.

\begin{table}[!t]\vspace{+0pt}
\centering
\caption{Key notations in our DNN inference latency and cost model in serverless platforms.}
\label{tab:notations}\vspace{-0pt}
\renewcommand{\arraystretch}{1.3}\vspace{+5pt}
\resizebox{\linewidth}{!}{ 
\begin{tabular}{c|p{5.6cm}}
\toprule[1pt]
\textbf{Notation}   & \textbf{Definition}                                   \\
\toprule[1pt]
\multirow{2}{*}{$L_{avg}^{t}, L_{max}^{t}$}             & Average, maximum inference latency of an inference workload on functions with a type $t=c,g$ \\ \hline
\multirow{2}{*}{$L_{0}^{g}$}                            & Inference latency on a GPU function configured with the maximum GPU memory size $M_{max}$ \\ \hline
\multirow{2}{*}{$\alpha_{b}^{avg}, \beta_{b}^{avg}, \gamma_{b}^{avg}$}    & Model coefficients of inference average latency on CPU functions with the batch size set as $b$\\ \hline
\multirow{2}{*}{$\alpha_{b}^{max}, \beta_{b}^{max}, \gamma_{b}^{max}$}    & Model coefficients of inference maximum latency on CPU functions with the batch size set as $b$\\ \hline
\multirow{2}{*}{$\xi_1, \xi_2$}                         & Model coefficients of inference latency on GPU functions \\ \hline
\multirow{2}{*}{$c, m$}                                 & vCPU cores of provisioned CPU functions, GPU memory of provisioned GPU functions\\ \hline
\multirow{2}{*}{$C^{\mathcal{X}}$}                           & Average monetary cost of an inference request of the application group $\mathcal{X}$\\ \hline
\multirow{2}{*}{$K_1, K_2, K_3$}                        & Unit cost of a vCPU core, GPU memory, and a function invocation\\ \hline
\multirow{2}{*}{$r^{\mathcal{X}}$, $b^{\mathcal{X}}$, $T^{\mathcal{X}}$}                            & Total request arrival rate, batch size, equivalent batching timeout of the application group $\mathcal{X}$\\
\bottomrule[1pt]
\end{tabular}
}
\vspace{-10pt}
\end{table}

\subsection{Modeling Latency of Serverless Inference}
\label{sec:model-latency}

We focus on two key metrics including the average inference latency and maximum inference latency. The former is used to calculate the monetary cost of inference, while the latter is used to evaluate SLO violations.

\textbf{CPU functions.} As elaborated in Sec.~\ref{sec:motivation-performance}, the average inference latency on CPU functions $L_{avg}^{c}$ decreases exponentially as more vCPU cores are provisioned. The average latency for batch size $b$ on CPU functions can be given by
\begin{equation}
    L_{avg}^{c} = \alpha_{b}^{avg}\cdot \exp\big(-\frac{c}{\beta_{b}^{avg}}\big) + \gamma_{b}^{avg},
\label{eq-cpu-lat}
\end{equation}
where $\alpha_{b}^{avg}, \beta_{b}^{avg}, \gamma_{b}^{avg}$ are the model coefficients with the batch size set as $b$, and the variable $c$ denotes the allocated vCPU cores of functions. We apply a similar method to model the maximum inference latency on CPU functions $L_{max}^{c}$ by leveraging the model coefficients $\alpha_{b}^{max}, \beta_{b}^{max}$ and $\gamma_{b}^{max}$.

\textbf{GPU functions.} When a GPU function is provisioned with the maximum GPU memory of $M_{max}$ (\emph{e.g.,} $24$ GB for an NVIDIA A10 GPU), the GPU function exclusively occupies a whole GPU device, leading to a \emph{stable} inference latency. As evidenced in Sec.~\ref{sec:motivation-performance}, the average and maximum latency overlap each other with the function GPU memory set as $M_{max}$. In addition, the inference latency on GPU functions with $M_{max}$ is roughly linear to the batch size $b$. Accordingly, we formulate the inference latency $L_{0}^{g}$ on GPU functions with $M_{max}$ as
\begin{equation}
    L_{0}^{g} = \xi_1 \cdot b + \xi_2,
\label{eq-gpu-lat}
\end{equation}
where $\xi_1$ and $\xi_2$ are model coefficients for GPU functions. 

We further model the average inference latency $L_{avg}^{g}$ and maximum inference latency $L_{max}^{g}$ on GPU functions. To facilitate serverless GPU functions, Alibaba Cloud Function Compute~\cite{Ali_function} deploys GPU temporal sharing mechanism (\emph{i.e.,} cGPU~\cite{cGPU}). Specifically, cGPU partitions the GPU's computing power into $M_{max}$ units with each lasting for a duration of $\tau$. It combines multiple time slices into a larger time slice $m \cdot \tau$ which is assigned to a GPU function with the GPU memory set as $m$. Though the latency of an inference request on GPU functions can be influenced by its arrival time, the average inference latency is still roughly linear to $L_{0}^{g}$ (in terms of the batch size), which is estimated as
\begin{equation}
    L_{avg}^{g} = \frac{M_{max}}{m} \cdot L_{0}^{g}.
\label{eq-gpu-avg-lat}
\end{equation}
Furthermore, by assuming an inference request demands $2m\cdot\tau$ to complete the execution, 
the maximum and minimum inference latency can be obtained as $2M_{max} \cdot \tau$ and $(M_{max} + m) \cdot \tau$, respectively, 
as shown in Fig.~\ref{model-timeslice}.
Accordingly, as for a request that demands $L_{0}^{g}$ to complete the execution and arrives at the start of the preempted time slice, it requires undergoing an additional number (\emph{i.e.,} $\big\lceil \frac{L_{0}^{g}}{m \cdot \tau} \big\rceil$) of preempted time slices. As a result, we formulate the maximum inference latency on GPU functions as
\begin{equation}
    L_{max}^{g} = \Big\lceil \frac{L_{0}^{g}}{m \cdot \tau} \Big\rceil \cdot (M_{max}-m) \cdot \tau + L_{0}^{g}.
\label{eq-gpu-max-lat}
\end{equation}

\begin{figure}[!t]
  	\centerline{
  		\includegraphics[width=0.92\linewidth]{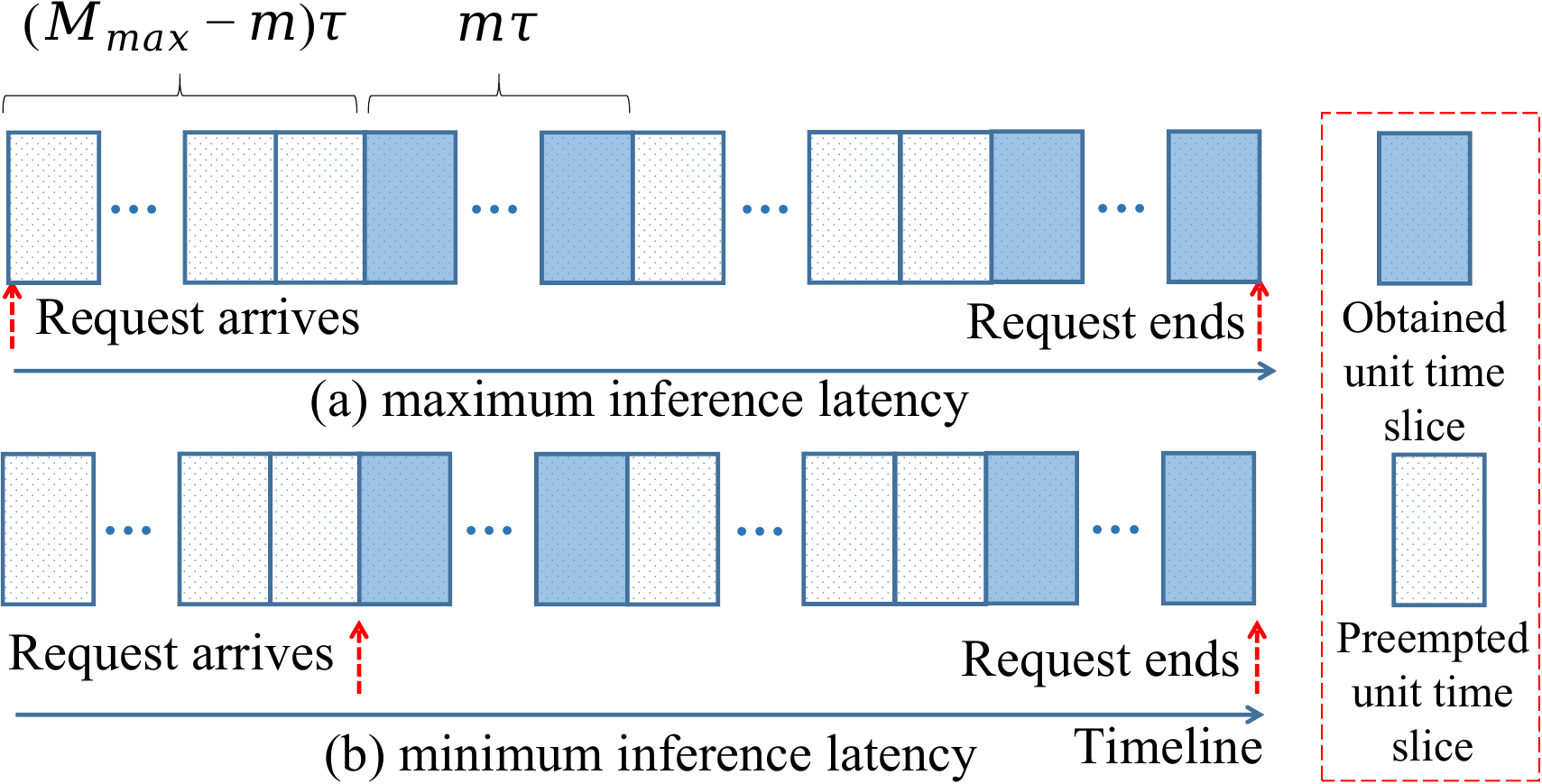}}
	\caption{Maximum and minimum inference latency scenarios caused by GPU time-slicing scheduling mechanism~\cite{cGPU} on a GPU function with the GPU memory set as $m$. It hosts an inference request that demands $2m \cdot \tau$ time slices, where $\tau$ denotes a unit GPU time slice. (a) The request arrives at the beginning of a \emph{preempted} time slice, resulting in the maximum inference latency $2M_{max} \cdot \tau$; and (b) the request arrives at the start of the \emph{obtained} time slice, leading to the minimum inference latency $(M_{max} + m) \cdot \tau$.}
	\label{model-timeslice}	\vspace{-10pt}
\end{figure}

\textbf{Model coefficients acquisition.} Based on our model above, we have six workload-specific coefficients (\emph{i.e.,} $\alpha_{b}^{avg}, \beta_{b}^{avg}, \gamma_{b}^{avg}, \alpha_{b}^{max}, \beta_{b}^{max}, \gamma_{b}^{max})$ for CPU functions, two workload-specific coefficients (\emph{i.e.,} $\xi_1, \xi_2$) and one hardware-specific parameter (\emph{i.e.,} $\tau$) for GPU functions. 
To determine the six coefficients for CPU functions, we execute DNN inference workloads $100$ times on each configuration by varying function vCPU cores and batch sizes. 
As evidenced by Sec.~\ref{sec:motivation-performance}, CPU functions typically outperform GPU functions with small batch sizes. To obtain model coefficients, we only profile DNN inference workloads with small batch sizes (\emph{i.e.,} ranging from $1$ to $4$), which significantly reduces the profiling overhead.

To obtain the workload-specific coefficients for GPU functions, we configure the GPU function with the GPU memory set as $M_{max}$ and execute DNN inference workloads only three times with two different batch sizes, as the inference latency on GPUs is stable.
To identify $\tau$, we get the $L_{max}^g$ and $L_0^g$ by running VGG-19 inference on a GPU function with a suitable GPU memory $m$ and $M_{max}$ for $100$ times, respectively.

\subsection{Modeling Monetary Cost of Multi-SLO DNN Inference}
\label{sec:model-cost}

In a scenario with multiple applications (\emph{i.e.,} a group $\mathcal{X}$), we assume that an inference application App$_i$ in the group $\mathcal{X}$ follows the Poisson distribution with the request arrival rate $r_i$. The application requests are first cached in a buffer with the capacity of $b^{\mathcal{X}}$ for batching. We set a batching timeout $T_i$ for each application to avoid long request waiting in the buffer. Once any $T_i$ expires, the cached requests (in a batch) are then sent to the functions for inference. 

To execute the inference with the maximum batch size (\emph{i.e.,} $b^{\mathcal{X}}$), the prerequisite $\lfloor r^{\mathcal{X}} \cdot T^{\mathcal{X}} \rfloor + 1 \geq b^{\mathcal{X}}$ needs to be held, where $\lfloor r^{\mathcal{X}} \cdot T^{\mathcal{X}} \rfloor + 1$ denotes the total number of requests received over a period $T^{\mathcal{X}}$, including the first request.
In more detail, $r^{\mathcal{X}}$ is the total request arrival rate, and $T^{\mathcal{X}}$ is the \emph{equivalent batching timeout} of group $\mathcal{X}$, which is considered as the \emph{expectation value} of request waiting time in the buffer, as each application has its own batching timeout. To illustrate that, we start from two applications with their request arrival rates and batching timeouts (\emph{i.e.,} App$_1$ with $r_1, T_1$ and App$_2$ with $r_2, T_2$).
By assuming that $T_1$ is smaller than $T_2$, the equivalent batching timeout of $\mathcal{X}$ can be calculated by
\begin{equation}
\begin{split}
    T^{\mathcal{X}} = T_1 + \eta_2 \cdot \frac{1-\exp\big(-r_1 \cdot (T_2-T_1)\big)}{r_1},
    \label{eq-t-exp}
\end{split}
\end{equation}
where $\eta_2 = \frac{r_2}{r_1+r_2}$ denote the probability of the first request from App$_2$. The derivation can be found in Appendix~\ref{sec:appendix-eq}~\cite{harmonybatch}.
To obtain the equivalent batching timeout of a large group $\mathcal{X}$ with two more applications, we can \emph{iteratively} apply Eq.~\eqref{eq-t-exp} to a sequence of two applications in the group $\mathcal{X}$.

According to the pricing of function resources in Alibaba Cloud Function Compute~\cite{Ali_function}, we further model the average monetary cost $C^{\mathcal{X}}$ of an inference request in terms of vCPU cores, GPU memory, and function invocations, which is given by
\begin{equation}
    C^{\mathcal{X}} = \frac{1}{b^{\mathcal{X}}} \big[ L_{avg}^{t} \cdot (c \cdot K_1 + m \cdot K_2) + K_3\big],
\label{eq-cost}
\end{equation}
where $L_{avg}^{t}$ is the average inference latency with the batch size $b^{\mathcal{X}}$ and the function type $t$. $K_1, K_2$ is the unit cost of vCPU cores $c$ and GPU memory $m$. $K_3$ is the constant cost of each function invocation. 
In particular, $m = 0$ represents a CPU function, and $c = 0$ represents a GPU function.

\section{Algorithm Design}
\label{sec:algorithm}

In this section, we first formulate the optimization problem of function resource provisioning for multi-SLO DNN inference. We then design and implement \Name{} to provide predictable performance for multi-SLO DNN inference with heterogeneous serverless functions.

\subsection{Optimizing Serverless Inference Resource Provisioning}
\label{sec:algorithm-problem}

We assume a set of DNN inference applications $\mathcal{W} = \{w_1, w_2, ...., w_n\}$ sharing the same DNN model with the inference latency SLOs $\mathcal{S} = \{s^{w_1}, s^{w_2}, ..., s^{w_n}\}$ and request arrival rates $\mathcal{R} = \{r^{w_1}, r^{w_2}, ..., r^{w_n}\}$. 
We categorize the application set $\mathcal{W}$ into several \emph{groups} $\mathcal{G} = \{\mathcal{X}_1, \mathcal{X}_2, ..., \mathcal{X}_m\}$. Each group $\mathcal{X} = \{w_j, w_{j+1}, ...\}$ is provisioned with an appropriate CPU function or GPU function, with the aim of meeting application SLO requirements while minimizing the average monetary cost of each inference request. 
Based on the DNN inference performance and cost models in Sec.~\ref{sec:model}, we can formulate the optimization problem as
\begin{eqnarray}
    \min_{\mathcal{G}, \mathcal{F}, \mathcal{B}} & & Cost = \sum_{\mathcal{X} \in \mathcal{G}} \eta^{\mathcal{X}} \cdot C^{\mathcal{X}}
    \label{eq-optimization}\\
    \text{s.t.}
    & & m^{\mathcal{X}} \geq M^{\mathcal{X}}, \  \forall \ \mathcal{X} \in \mathcal{G} \label{eq-cons-gmem}\\
    & & b^{\mathcal{X}} \leq \lfloor r^{\mathcal{X}} \cdot T^{\mathcal{X}} \rfloor + 1, \forall \ \mathcal{X} \in \mathcal{G} \label{eq-cons-rtb}\\
    & & t^w + L_{max}^{t} \leq s^w, \  \forall \ w \in \mathcal{X}, \ \mathcal{X} \in \mathcal{G} \label{eq-cons-SLO}
\end{eqnarray}
where $\eta^{\mathcal{X}}$ is the ratio of the request arrival rate of a group $\mathcal{X}$ to the total request arrival rate. $C^{\mathcal{X}}$ is the average monetary cost of a group $\mathcal{X}$.
Each group $\mathcal{X}$ is configured with a function of resource $f^{\mathcal{X}} \in \mathcal{F}$ (\emph{i.e.,} a tuple of vCPU cores $c^{\mathcal{X}}$ and GPU memory $m^{\mathcal{X}}$, $f^{\mathcal{X}} = [c^{\mathcal{X}}, m^{\mathcal{X}}]$).
Constraint~\eqref{eq-cons-gmem} guarantees the GPU memory demands $M^{\mathcal{X}}$ of inference, which are proportional to the batch size.
In addition, $b^{\mathcal{X}} \in \mathcal{B}$ denotes the batch size configured to group $\mathcal{X}$ and $t^w$ is the timeout configured with the application $w$.
Constraint~\eqref{eq-cons-rtb} guarantees that DNN inference is executed with the batch size $b^{\mathcal{X}}$.
Constraint~\eqref{eq-cons-SLO} guarantees the latency SLO $s^{w}$ for an application $w$, where the $L_{max}^{t}$ is the maximum inference latency with the batch size $b^{\mathcal{X}}$. $T^{\mathcal{X}}$ is calculated by $t^{w}$ in the group $\mathcal{X}$ by Eq.~\eqref{eq-t-exp}. To greedily enlarge the batching timeout $T^{\mathcal{X}}$, we set the $t^{w}$ as the maximum value which meets the Constraint~\eqref{eq-cons-SLO} (\emph{i.e.,} $t^{w} = s^w - L_{max}^{t}$).

\textbf{Problem analysis.} 
Our group solution $\mathcal{G}$ has a large searching space of $B_{|\mathcal{W}|}$, which is the Bell number~\cite{bell1938iterated}. Given a group $\mathcal{G}$, the function resource $f \in \mathcal{F}$ can only take limited discrete values. Meanwhile, the batch size $b \in \mathcal{B}$ is constrained to integer values. Accordingly, the resource provisioning problem can be reduced to an integer programming problem.
Obviously, the total average monetary cost and Constraint~\eqref{eq-cons-SLO} is non-linear with the configuration parameters, and thus our optimization problem can further be reduced to a \emph{non-linear integer programming} problem, which is an NP-hard problem~\cite{boyd2004convex}. We turn to designing a heuristic algorithm in Sec.~\ref{sec:algorithm-design} to solve such an optimization problem.

\SetAlFnt{\small}
\SetAlgoVlined \vspace{-0pt}
\begin{algorithm}[!t]
\caption{\Name{}: Two-stage merging strategy for application groups.}
\label{config-alg-group}
\SetKwInOut{KwIn}{Input}
\SetKwInOut{KwOut}{Output}
\KwIn{A set of applications $\mathcal{W}$ with their SLOs $\mathcal{S}$ and arrival rates $\mathcal{R}$.}
\KwOut{A set of group $\mathcal{G}$, sets of function provisioning plans $\mathcal{F}$ and batch size $\mathcal{B}$.}
\BlankLine
\textbf{Initialize:} $\mathcal{G} \gets \{\{w_1\}, \{w_2\}, ..., \{w_n\}\}$\;
Provision function resources for each $\mathcal{X} \in \mathcal{G}$, $C^{\mathcal{X}}, f^{\mathcal{X}}, b^{\mathcal{X}} \gets \mathtt{funcProvision}$($\mathcal{X}, s^{\mathcal{X}}, r^{\mathcal{X}}$)\;
Sort the applications in $\mathcal{G}$ with SLOs and initialize the group list $\mathbf{L} \gets sortWithSLO(\mathcal{G})$\;
\tcp{Stage1: merging CPU functions}
Set index $i \gets 0, j \gets 0$; Set the request arrival rate $r \gets 0$\; 
\While {$i < |\mathbf{L}|$} {
    \If{$c^{\mathbf{L}[i]} > 0$} { 
        $r \gets r + r^{\mathbf{L}[i]}$\;
        \tcp{$r^*$ is the arrival rate knee point}
        \If{$r > r^*$} {
            $\mathbf{L}, \_ \gets$ \texttt{Merge}($\mathbf{L}, j, i+1$)\;
            Set $i \gets j$, $j \gets j+1$ and $r \gets 0$\;
        }
    }
    \Else {
        Set $j \gets i+1$ and $r \gets 0$\;
    }
    Set $i \gets i + 1$\;
}

\tcp{Stage2: merging GPU functions}
Set index $i \gets 0$\;
\While {$i < |\mathbf{L}| - 1$} {
    \If {$m^{\mathbf{L}[i]} > 0$ or $m^{\mathbf{L}[i+1]} > 0$} { 
        $\mathbf{L}, isMerged \gets$ \texttt{Merge}($\mathbf{L}, i, i+2$)\;
        \If{isMerged} {
            $i \gets i-1$\;
        }
    }
    Set $i \gets i + 1$\;
}
\Return {$\mathcal{G}, \mathcal{F}, \mathcal{B}$}\;

\BlankLine
\SetKwFunction{FMerge}{Merge}
\SetKwProg{Fn}{Function}{\string :}{}
\Fn{\FMerge{$\mathbf{L}$, low, high}} {
    $\mathcal{X} \gets \mathbf{L}[low] \cup \mathbf{L}[low+1] \cup ...\mathbf{L}[high-1]$\;
    Provision function resources for group $\mathcal{X}$, $C^{\mathcal{X}}, f^{\mathcal{X}}, b^{\mathcal{X}} \gets \mathtt{funcProvision}(\mathcal{X}, s^{\mathcal{X}}, r^{\mathcal{X}})$\;
    \If {$C^{\mathcal{X}}$ is lower than the cost before merging} {
        Update $\mathcal{G}$, $\mathcal{F}$ and $\mathcal{B}$ with the function provisioning plan\;
        $\mathbf{L} \gets \mathbf{L}[:low] + \mathcal{X} + \mathbf{L}[:high]$\;
        \Return {$\mathbf{L}$, True}\;
    }
    \Return {$\mathbf{L}$, False}\;
}
\end{algorithm}

\subsection{Design of \emph{\Name{}} Strategy}
\label{sec:algorithm-design}

\Name{} divides the problem into two parts: \emph{First} is to divide the applications into different groups, and \emph{second} is to provision function resources for each application group.

\textbf{Two-stage merging for application groups.}
We consider placing the adjacent applications sorted by their SLOs in \emph{ascending order} into a group. According to Constraint~\eqref{eq-cons-SLO}, the batching timeout gap between two applications can be substantial if their SLO difference is large.
If the two applications (with batching timeouts denoted as $T_1$ and $T_2$, where $T_1 < T_2$) are grouped together, the equivalent batching timeout $T^{\mathcal{X}}$ becomes much less than $T_2$ by analyzing Eq.~\eqref{eq-t-exp}. This causes a smaller \emph{aggregated} batch size of DNN inference for all groups $\mathcal{X} \in \mathcal{G}$, thereby leading to a lower cost-efficient function resource provisioning, as evidenced by Sec.~\ref{sec:motivation-performance}.

Based on our analysis above, we design a two-stage group merging strategy in Alg.~\ref{config-alg-group}. 
Initially, each application is considered as a group, and the function provisioning plan and monetary cost for each group are calculated by our $\mathtt{funcProvision}$ strategy (lines $1$-$2$).
In the \emph{first} stage, groups originally deployed on CPU functions are merged to be deployed on GPU functions as much as possible.
To support the group merging for applications with adjacent SLOs, we sort the applications in $\mathcal{G}$ into a list $\mathbf{L}$ based on their SLOs, and only consecutive groups (or applications) in the list are merged (line $3$).
We leverage the \emph{knee point} of the request arrival rate as illustrated in Fig.~\ref{motivation-rps} to be the \emph{threshold} $r^{*}$ for group merging. If the total request arrival rate exceeds $r^{*}$, merging provides an opportunity to configure a more efficient GPU function (lines $4$-$13$).
In the \emph{second} stage, groups deployed on GPU functions and adjacent groups based on SLOs are merged as much as possible to increase the request arrival rate of the merged groups.
By iterating each group on GPU functions, \Name{} examines whether they can reduce the monetary cost after group merging (lines $14$-$20$). After that, \Name{} outputs the application groups and their corresponding function provisioning plans.

\textbf{$\mathtt{funcProvision}$: Function provisioning for an application group.}
We analyze the optimization of CPU/GPU function resource provisioning.
The configuration space of vCPU cores $c^{\mathcal{X}} \in [0.05, 16]$ for CPU functions with the step of $0.05$ is larger than that of $b^{\mathcal{X}} \in [1,4]$. The configuration space of GPU memory $m^{\mathcal{X}} \in [1, 24]$ for GPU functions with the step of $1$ is smaller than that of $b^{\mathcal{X}} \in [1, 32]$. We find that the CPU function exhibits a smaller batch size configuration space ($4$ choices) but a larger resource configuration space ($320$ choices), whereas the GPU function shows the opposite characteristics. To speed up searching for the optimal solution, we derive two theorems by analyzing Eq.~\eqref{eq-cost}.

\newtheorem{theorem}{Theorem}
\begin{theorem}\label{thm-cpu}
    Given a batch size $b^{\mathcal{X}}$, the minimum cost of CPU function provisioning can be achieved with the allocated vCPU cores set as $c^*$ (i.e., the relative minimum point or the boundary points).
\end{theorem}
\begin{proof}
    The proof can be found in Appendix~\ref{sec:appendix-theorem1}~\cite{harmonybatch}.
\end{proof}

\begin{theorem}\label{thm-gpu}
    Given an amount of GPU memory $m^{\mathcal{X}}$, the minimum cost of GPU function provisioning can be achieved if the following condition holds. 
    \begin{equation}
        \lfloor r^{\mathcal{X}} \cdot T^{\mathcal{X}} \rfloor + 1 = b^{\mathcal{X}}.
        \label{eq-cons-eq-rt-b}
    \end{equation}
\end{theorem}
\begin{proof}
    The proof can be found in Appendix~\ref{sec:appendix-theorem2}~\cite{harmonybatch}.
\end{proof}


Based on Theorem~\ref{thm-cpu} and Theorem~\ref{thm-gpu}, we simply adopt the \emph{binary search} method in \Name{} to fast identify the cost-efficient function provisioning plan (\emph{i.e.,} the vCPU cores for the CPU function and the batch size for the GPU function, respectively). 
Accordingly, \Name{} can minimize the inference budget, while guaranteeing the latency SLOs for an application group $\mathcal{X}$.

\textbf{Remark.} The complexity of Alg.~\ref{config-alg-group} is in the order of $\mathcal{O}(|\mathcal{W}| \cdot M_{max} \cdot \log_{2}B_{max})$, where $|\mathcal{W}|$ is the number of the applications. $M_{max}$ and $B_{max}$ are the maximum GPU memory and the maximum batch size, respectively.
This is attributed to the relatively small batch size considered by the CPU function, which makes the complexity of the function provision strategy mainly depend on the size of the search space of GPU function resource provisioning. The computation overhead of Alg.~\ref{config-alg-group} is well contained, which will be evaluated in Sec.~\ref{sec:eval-overhead}.

\subsection{Implementation of \emph{\Name{}} Prototype}
\label{sec:algorithm-prototype}

The \Name{} prototype is implemented on the Alibaba Compute Function platform~\cite{Ali_function} with over $1,400$ lines of Python codes, which are publicly available on GitHub. 
We implement four representative DNN inference workloads based on 
PyTorch\footnote{https://pytorch.org} v1.13.0 and ONNX Runtime\footnote{https://onnxruntime.ai} v1.16.1. We use the Alibaba FC-Open Python SDK~\cite{Ali_function} to update function resources.
We deploy the \Name{} on a dedicated cloud instance, which receives DNN inference requests from a set of user applications. We first set up a request queue for each application group. We then batch the inference requests in the queue and finally route them to the provisioned CPU/GPU functions. To handle the request arrival variations, \Name{} can be \emph{periodically} executed to provision functions for DNN inference workloads. \Name{} mainly determines the batching and function resource configurations. The batching configurations are sent to the \emph{batch manager} to control the request queue, while the function resource configurations are sent to the serverless platform to \emph{vertically} scale up or scale down functions.


\section{Performance Evaluation}
\label{sec:evaluation}

In this section, we evaluate \Name{} by conducting a set of prototype experiments with four representative DNN models (listed in Table~\ref{tab:evaluation-workloads}) deployed on Alibaba Cloud Function Compute~\cite{Ali_function}. We seek to answer the questions as follows.
\begin{itemize}
    \item \textbf{Accurary:} Can our model in \Name{} accurately predict the DNN inference latency with heterogeneous serverless functions? (Sec.~\ref{sec:eval-predict})
    \item \textbf{Effectiveness:} Can our function provisioning strategy in \Name{} provide predictable performance for multi-SLO DNN inference while saving the monetary cost? (Sec.~\ref{sec:eval-effectiveness})
    \item \textbf{Overhead:} How much runtime overhead does \Name{} practically bring? (Sec.~\ref{sec:eval-overhead})
\end{itemize}

\begin{figure*}[!t]   
	\begin{minipage}[t]{0.50\linewidth}
  	\centerline{
  		\subfigure[VideoMAE]{\includegraphics[width=0.49\linewidth]{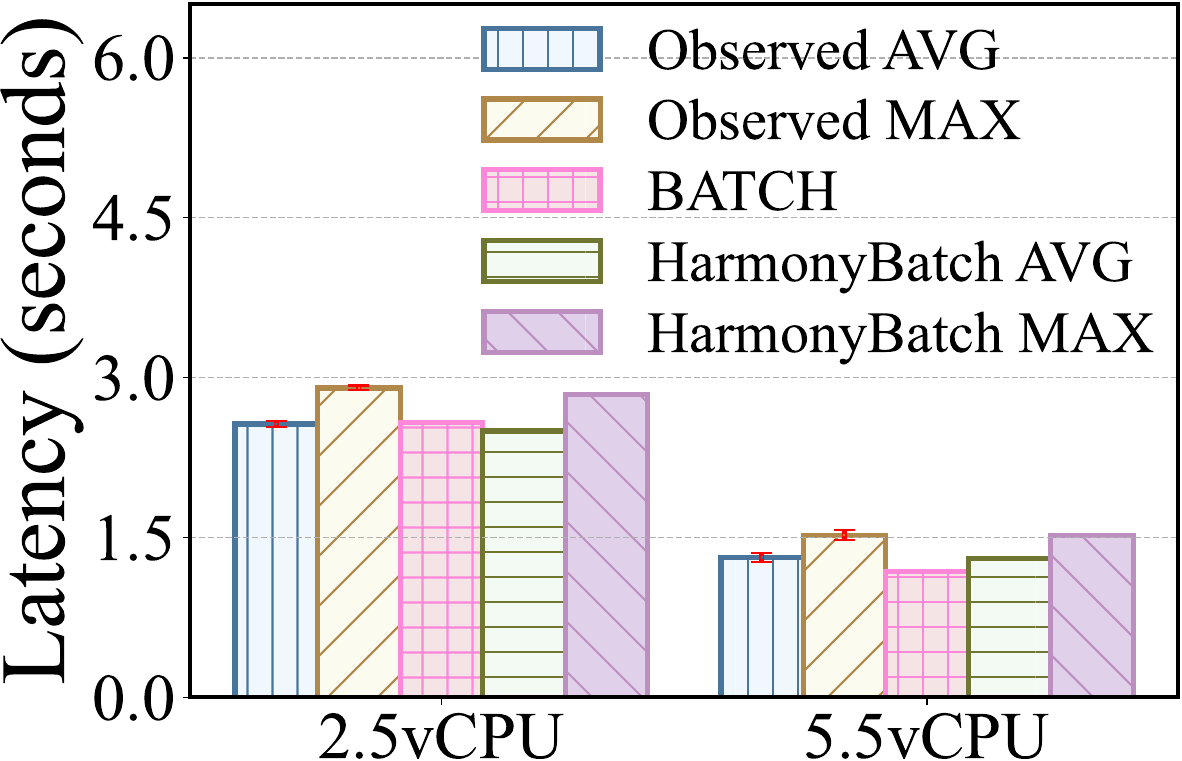}\label{evaluation-accuracy-vediomae}}
  		\subfigure[VGG-19]{\includegraphics[width=0.49\linewidth]{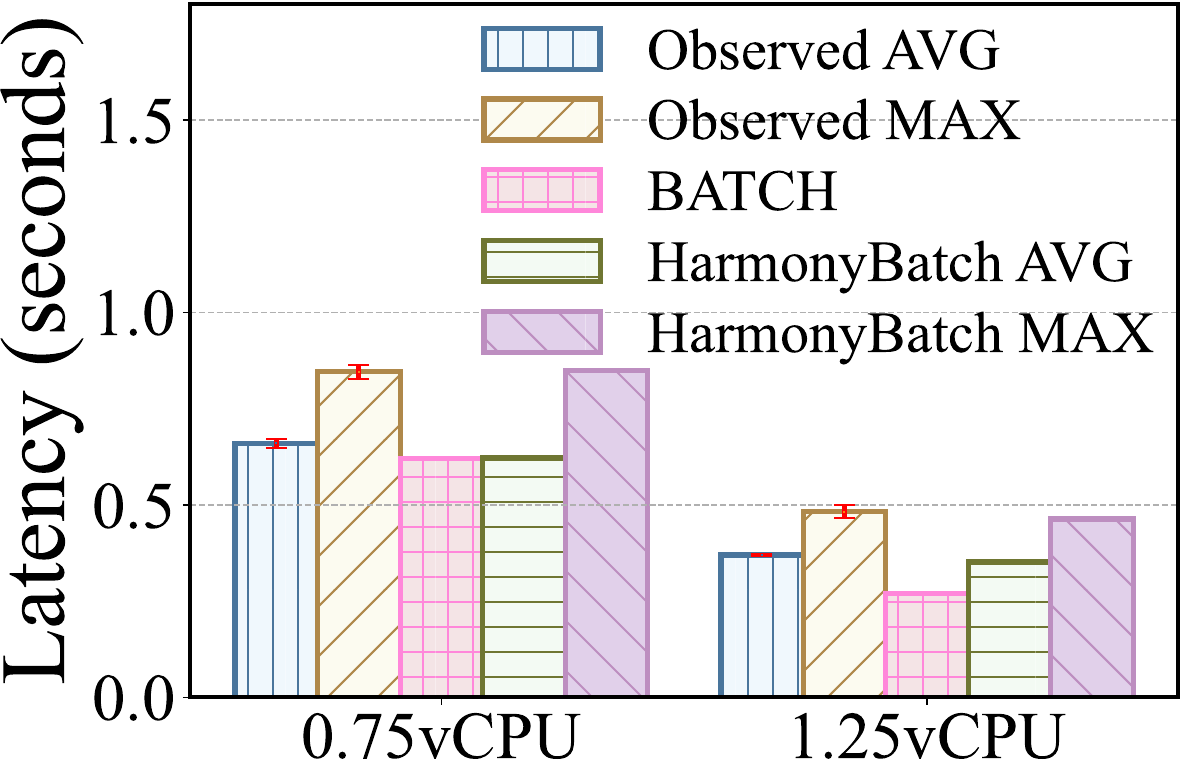}\label{evaluation-accuracy-vgg19}}}
	\caption{Comparison of the observed and predicted inference latency of VideoMAE and VGG-19 executed on CPU functions.}
	\label{evaluation-accuracy-cpu}
	\end{minipage} \hspace{+4pt}
 	\begin{minipage}[t]{0.50\linewidth}
  	\centerline{
  		\subfigure[BERT]{\includegraphics[width=0.49\linewidth]{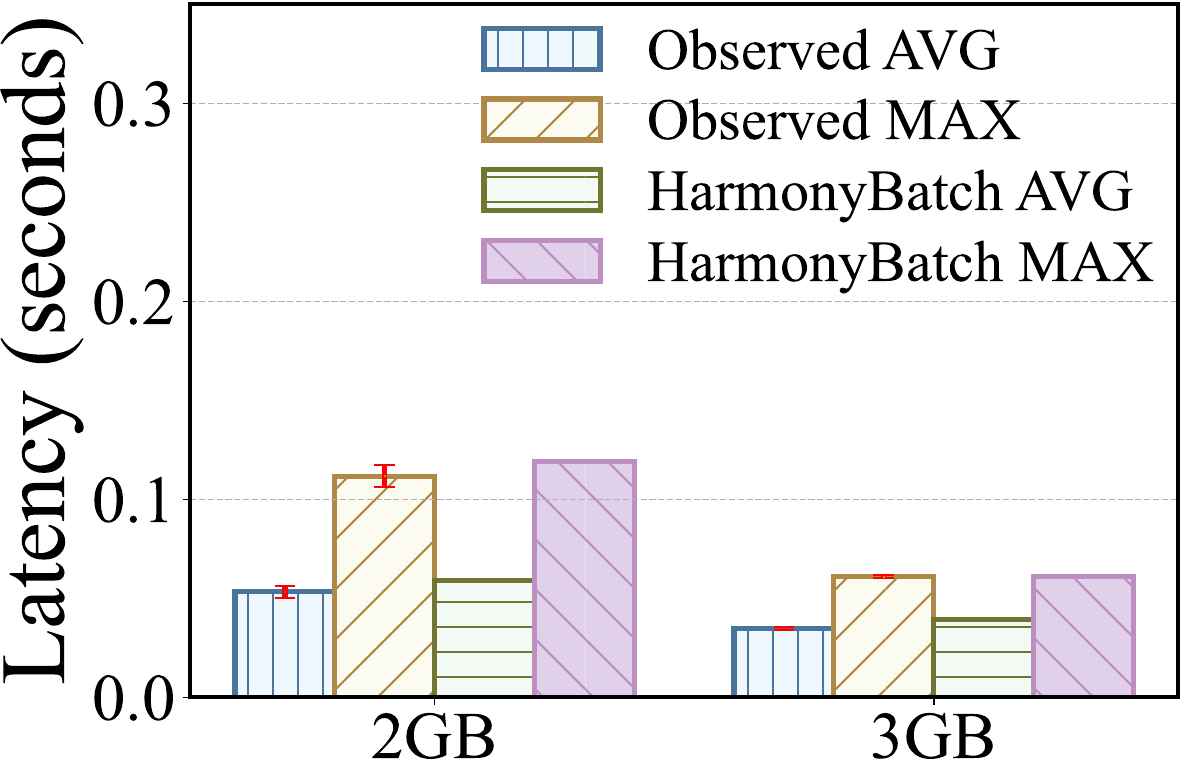}\label{evaluation-accuracy-BERT-gpu}}
            \subfigure[GPT-2]{\includegraphics[width=0.49\linewidth]{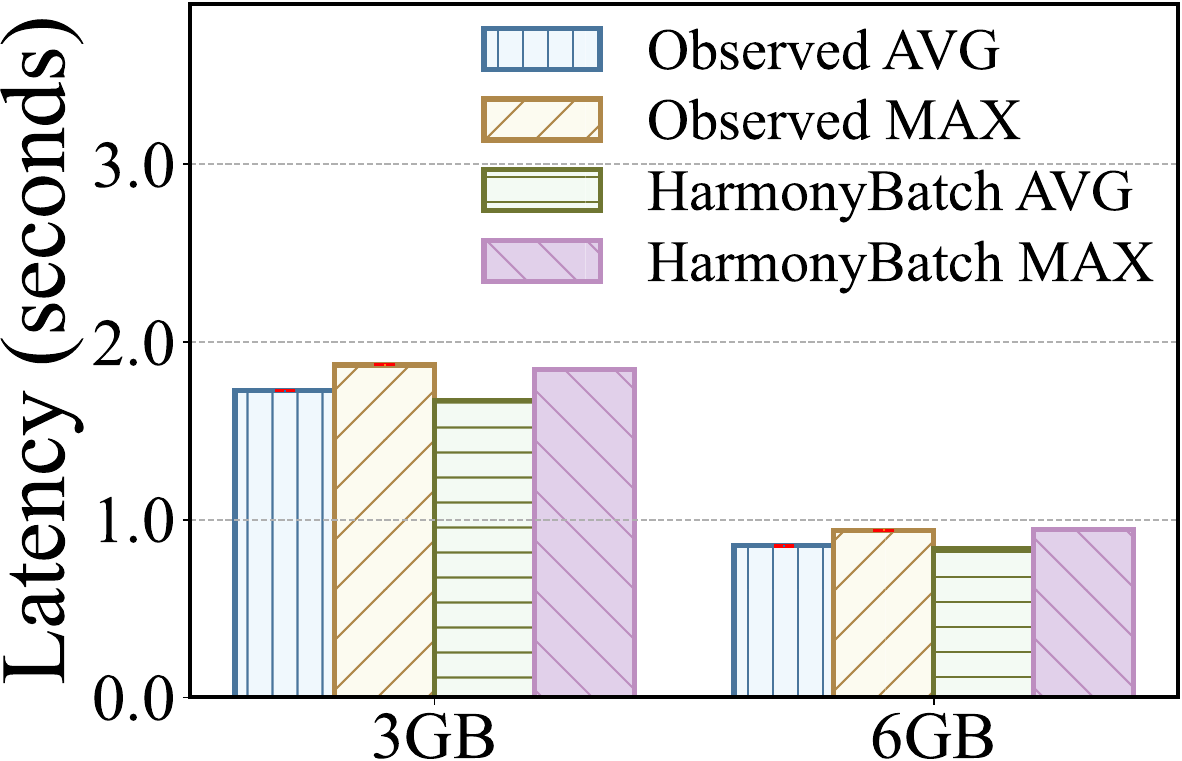}\label{evaluation-accuracy-gpt2-gpu}}}
	\caption{Comparison of the observed and predicted latency of BERT and GPT-2 executed on GPU functions. BATCH does support GPU functions.}
	\label{evaluation-accuracy-gpu}
	\end{minipage}
	\vspace{-12pt}
\end{figure*}

\begin{figure*}\vspace{+6pt}
	\begin{minipage}[t]{0.32\linewidth}
		\centering
		\includegraphics[width=2.3in]{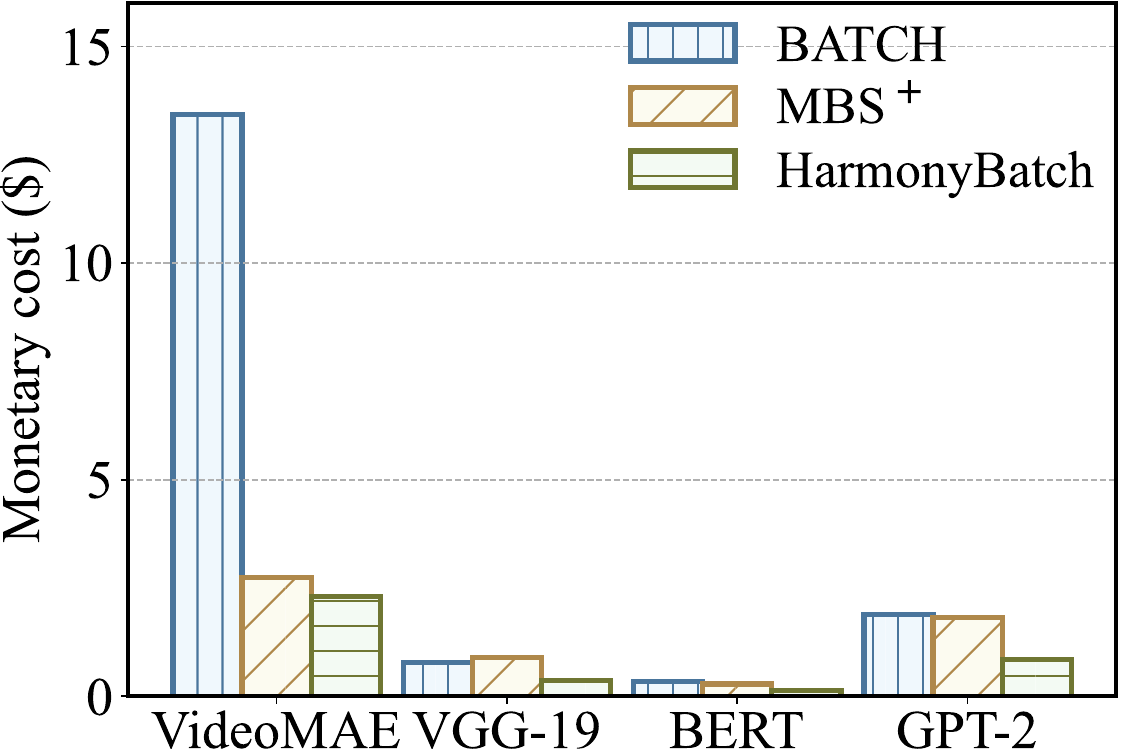}
		\caption{Comparison of the monetary cost of various function resource provisioning strategies for representative DNN models.}
		\label{evaluation-cost-trace}
	\end{minipage}\hspace{+4pt}
        \begin{minipage}[t]{0.32\linewidth}
		\centering
            \includegraphics[width=2.3in]{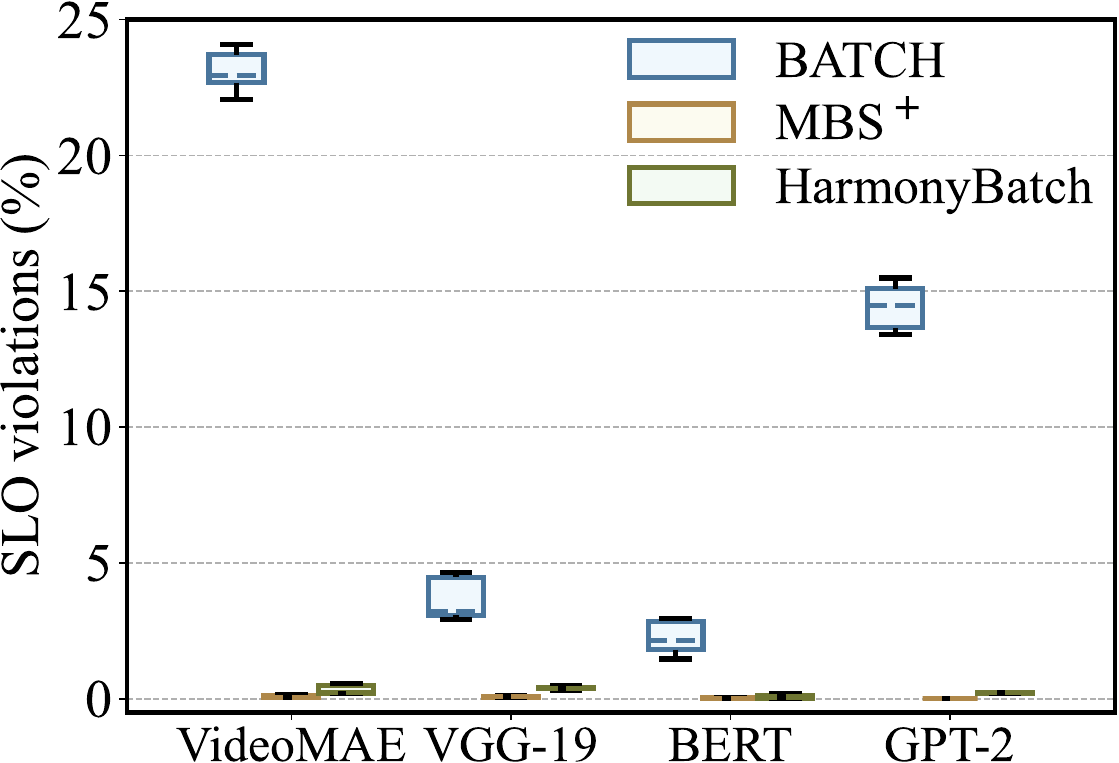}
		\caption{Comparison of the SLO violations of various function resource provisioning strategies for representative DNN models.}
		\label{evaluation-slo}
	\end{minipage}\hspace{+4pt}
	\begin{minipage}[t]{0.32\linewidth}
		\centering
  		\includegraphics[width=2.26in]{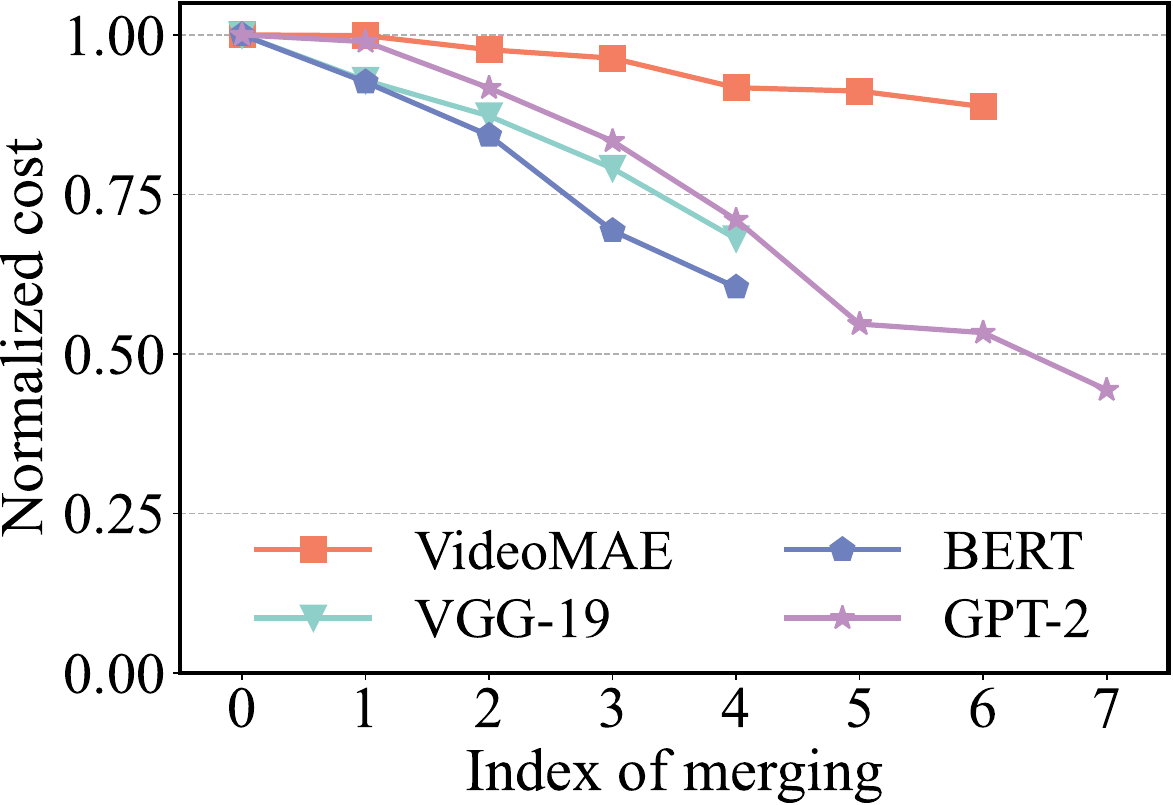}
		\caption{Normalized monetary cost of four representative DNN models by increasing the merging index of application groups over time.}
		\label{evaluation-merging}	
	\end{minipage}
	\vspace{-12pt}
\end{figure*}

\subsection{Experimental Setup}
\label{sec:eval-setup}

\textbf{Configurations of DNN inference workloads.} We select four representative DNN models as listed in Table~\ref{tab:evaluation-workloads} and replay the real-world trace from Azure Function~\cite{shahrad2020serverless} to evaluate the effectiveness of \Name{}. The DNN models are selected from diverse fields, \emph{i.e.,} VideoMAE~\cite{tong2022videomae} in video analytics, VGG-19~\cite{simonyan2014very} in image processing, BERT~\cite{devlin2019BERT} and GPT-2~\cite{radford2019language} in NLP. We use several widely-used datasets including Kinetics-400\footnote{https://www.deepmind.com/open-source/kinetics},
ImageNet\footnote{https://image-net.org/challenges/LSVRC/2017},
Wikipedia\footnote{https://en.wikipedia.org/wiki/English\_Wikipedia} and ShareGPT\footnote{https://sharegpt.com} for serving the four models above, respectively.

\textbf{Configurations of serverless functions.} We deploy our inference models in China Shanghai region of Alibaba Cloud Function Compute~\cite{Ali_function}. We adopt an \texttt{\Instance{}} ECS instance to deploy \Name{}. During the period of our experiments (Nov. 2023), the unit price of vCPU cores is $K_1 = 1.3e^{-5}$\$/vCPU$\cdot$s, and the unit price of GPU memory is $K_2 = 1.5e^{-5}$\$/GB$\cdot$s, as well as the constant unit cost of a function invocation is $K_3 = 1.3e^{-7}$\$.

\begin{table}[!t]\vspace{+0pt}
\renewcommand{\arraystretch}{1.5}
\vspace{+5pt}
\centering \caption{DNN inference workloads deployed in our experiments.} 
\label{tab:evaluation-workloads}\vspace{-0pt}
\resizebox{1.0\columnwidth}{!}{
\large
\begin{tabular}{c|cccc}
\toprule[1.2pt]
Workloads & VideoMAE & VGG-19 & BERT & GPT-2 \\
\midrule[1.2pt]
Framework & PyTorch & PyTorch & Onnx Runtime & PyTorch \\ \hline
Domains & Video Analytics & Image Processing & NLP & NLP \\ \hline
Datasets & Kinetics-400 & ImageNet & Wikipedia & ShareGPT \\
\bottomrule[1.2pt]
\end{tabular}
}
\vspace{-10pt}
\end{table}

\textbf{Baselines and metrics.} We compare \Name{} with two baselines: (1) BATCH~\cite{ali2020batch}: It leverages multi-variable parametric regression to model the latency, and separately provisions CPU functions only for each application using an exhaustive search. 
(2) MBS$^+$: 
It is the extended MBS~\cite{ali2022optimizing}, which employs Bayesian-based provisioning to distribute requests into several groups evenly. We incorporate our performance model into MBS$^+$ to support heterogeneous function provisioning for application groups. 
We focus on three metrics: \emph{SLO violations}, \emph{monetary cost}, and \emph{runtime overhead}.

\subsection{Validating Inference Performance Model of \emph{\Name{}}}
\label{sec:eval-predict}

\textbf{Can \Name{} well predict the DNN inference latency?} 
We first evaluate the latency of VideoMAE and VGG-19 on the CPU function by setting the batch size as $1$.
As depicted in Fig.~\ref{evaluation-accuracy-cpu}, \Name{} can well predict the average and maximum inference latency of VideoMAE and VGG-19 with the prediction error ranging from $0.2\%$ to $6.1\%$.
In contrast, BATCH~\cite{ali2020batch} poorly predicts the DNN inference latency as it treats the model latency as a deterministic distribution with the prediction error ranging from $0.7\%$ to $43.0\%$.
We next evaluate the inference latency of BERT and GPT-2 on the GPU function by setting the batch size as $8$.
As depicted in Fig.~\ref{evaluation-accuracy-gpu}, \Name{} exhibits strong predictive capabilities, accurately predicting the average and maximum inference latency of BERT and GPT-2 with prediction errors ranging from $0.1\%$ to $11.4\%$. In particular, the accuracy in predicting the maximum latency of BERT achieves a low prediction error of $0.1\%$, which can be attributed to the integration of GPU time-slicing scheduling mechanism into our inference performance model.

\subsection{Effectiveness of \emph{\Name{}} Function Provisioning}
\label{sec:eval-effectiveness} 

\textbf{Can \Name{} guarantee the DNN inference performance while minimizing the monetary cost?} 
To evaluate the efficacy of \Name{}, we utilize $8$ applications for each DNN model ($32$ applications in total). Specifically, we set the application SLOs between $0.2$ and $1.0$ seconds with an interval of $0.1$ seconds for VGG-19 and BERT. Also, we set the application SLOs between $1.0$ and $2.4$ seconds with an interval of $0.2$ seconds for VideoMAE and GPT-2.

As shown in Fig.~\ref{evaluation-cost-trace}, \Name{} can save the cost by up to \CostReduce{} compared with the two baselines. 
Specifically, BATCH can hardly reduce the cost with predictable performance, because it only batches requests for individual applications. Meanwhile, it overlooks the significant cost reduction of GPU functions, particularly for resource-intensive DNN models like VideoMAE.
In contrast, MBS$^+$ and \Name{} implement a heterogeneous serverless provisioning strategy, which leverages GPU functions to significantly enhance the advantages of batching, thereby gaining a substantial cost benefit.
Furthermore, \Name{} achieves cost reductions of $16.1\%$, $59.6\%$, $54.5\%$ and $52.9\%$ for the four workloads compared to MBS$^+$. This is attributed to the fact that MBS$^+$, with its evenly distributed requests, commonly aggregates inference requests with significant differences in SLOs into application groups. Moreover, it distributes the inference requests from the application with high request arrival rates to several functions, resulting in smaller batch sizes.

\begin{figure*}
	\begin{minipage}[t]{0.62\linewidth}
		\centering
		\includegraphics[width=4.4in]{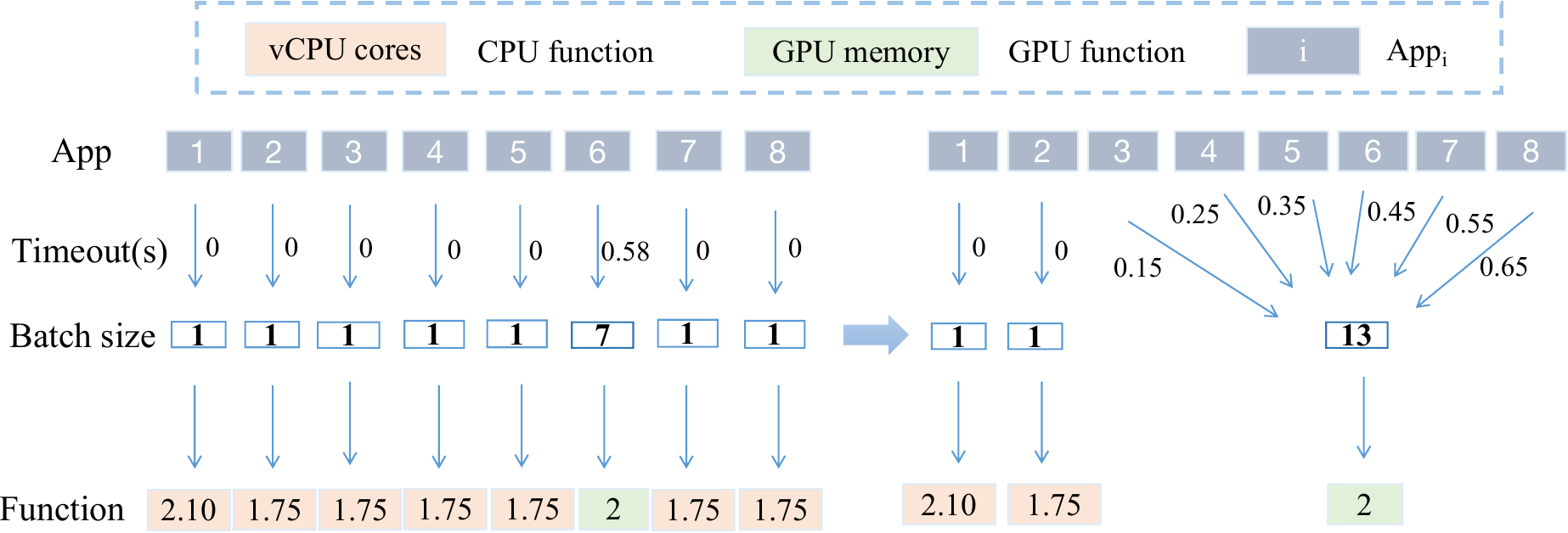}
		\caption{Function provisioning plans of VGG-19 before and after the group merging.}
		\label{evaluation-before_after}\vspace{-12pt}
	\end{minipage}\hspace{+4pt}
	\begin{minipage}[t]{0.36\linewidth}
		\vspace{-95pt}
		\centering
		\renewcommand{\arraystretch}{1.2}
		\resizebox{2.8in}{!}{
		\begin{tabular}{c|ccc}
		\toprule[1.0pt]
		Number of & \multirow{2}{*}{BATCH} & \multirow{2}{*}{MBS$^+$} & \multirow{2}{*}{\Name{}} \\
            applications & & & \\
		\midrule[1.0pt]
            $1$ & $34$ & $176$ & $2$ \\ \hline
            $6$ & $222$ & $3,497$ & $23$ \\ \hline
            $12$ & $393$ & $8,156$ & $38$ \\
		\bottomrule[1.0pt]
		\end{tabular}
		}
		\vspace{+6pt}
		\captionof{table}{Computation time (in milliseconds) of different strategies.}
		\label{tab:eval-overhead}\vspace{-12pt}
	\end{minipage}
\end{figure*}

Regarding the SLO violations as illustrated in Fig.~\ref{evaluation-slo}, BATCH can hardly provide predictable inference performance for DNN models, because it treats the inference request latency as a deterministic distribution, which leads to an excessively large SLO violations across large DNN models. In contrast, both \Name{} and MBS$^+$ can guarantee inference performance for all DNN models. Furthermore, \Name{} can obtain the function provisioning plans much faster than MBS$^+$ as elaborated in Sec.~\ref{sec:eval-overhead}.

\textbf{Can \Name{} reduce the monetary cost as the merging of application groups proceeds?} 
As depicted in Fig.~\ref{evaluation-merging}, \Name{} separately provisions function resources for each application at the initial time.
After $4$ to $7$ merging operations, the four DNN inference workloads achieve notable cost reduction by $13.1\%-62.4\%$. By taking VGG-19 as an example, it experiences a significant reduction in the monetary cost by up to $38.4\%$. The reason is that before the merging, the request rates of $7$ applications were low, making them be deployed on CPU functions. However, after the merging, $6$ applications are grouped and deployed on a GPU function with a large batch size, achieving a significant cost reduction while guaranteeing application SLOs.

We look into the adjustments on function provisioning plans made by \Name{} for VGG-19 during the two-stage merging process. As shown in Fig.~\ref{evaluation-before_after}, \Name{} provisions $7$ CPU functions and $1$ GPU function at first. After the merging process, \Name{} decreases the application groups as $3$ and optimizes the function provisioning to only $2$ CPU functions and $1$ GPU function. It directs $79.0\%$ of requests to the GPU function for DNN inference. Meanwhile, the batch size of the GPU function is increased from $7$ to $13$ after the merging process. As a result, \Name{} prioritizes assigning inference requests to GPU functions whenever feasible, with the aim of greedily creating large batches and thus significantly reducing the DNN inference budget.

\subsection{Runtime Overhead of \emph{\Name{}}}
\label{sec:eval-overhead}

We evaluate the runtime overhead of \Name{} including the workload profiling time and the computation time of Alg.~\ref{config-alg-group}. 
Specifically, the profiling time for obtaining model coefficients of VideoMAE~\cite{tong2022videomae}, VGG-19~\cite{simonyan2014very}, BERT~\cite{devlin2019BERT}, and GPT-2~\cite{radford2019language} on CPU functions are $13$, $5$, $11$, and $18$ minutes, respectively. The profiling time for each model listed in Table~\ref{tab:evaluation-workloads} is less than $1$ minute on GPU functions. 
In addition, the profiling time to obtain the minimum time slice $\tau$ is merely $0.1$ minutes which requires profiling only once. 

To evaluate the computation time of \Name{}, we provision a VGG-19 model with different numbers of applications from $1$ to $12$. As listed in Table~\ref{tab:eval-overhead}, the algorithm computation time of \Name{} is roughly linear to the number of applications, which is negligible as compared to the other two strategies. This is because we adopt a two-stage merging grouping strategy, which significantly reduces the complexity of the application grouping and function provisioning algorithms to $\mathcal{O}(|\mathcal{W}| \cdot M_{max} \cdot \log_{2}B_{max})$.
In particular, we evaluate the computation overhead of $\mathtt{funcProvision}$ for an application group.
Its computation time is still acceptable because we obtain function resource provision plan using the binary search method instead of an exhaustive search.

\section{Related Work}
\label{sec:related}

\textbf{Optimizing DNN inference with serverless functions.} To reduce the serverless inference budget, BATCH~\cite{ali2020batch} introduces batching inference requests in serverless platforms. MBS~\cite{ali2022optimizing} further optimizes the padding overhead by aggregating similar-size requests in a batch. 
To reduce the memory consumption, Tetris~\cite{li2022tetris} combines batching and concurrent executions in a serverless inference system. 
Different from prior works above, \Name{} aims to adequately configure the batch size for multi-SLO DNN inference workloads
by considering their SLOs explicitly.
INFless~\cite{yang2022infless} develops a serverless inference system with CPU and GPU resources by unifying their computing power using the floating point operations per second (FLOPS) metric. In contrast, \Name{} builds an analytical model for the inference performance and cost of public heterogeneous functions.
Moreover, \Name{} can benefit from several recent DNN inference optimizations such as scalability improvements (\emph{e.g.,} MArk~\cite{zhang2019mark}, AsyFunc~\cite{pei2023asyfunc}) and fine-grained selective batching techniques (\emph{e.g.,} Orca~\cite{yu2022orca}).


\textbf{Resource provisioning of serverless functions.} To optimize the function provisioning, AWS Lambda Power Tuning~\cite{aws-lambda-power-tuning} adopts workload profiling with all possible memory configurations to identify the optimal memory allocation for functions, which brings heavy overhead.
Sizeless~\cite{eismann2021sizeless} adopts a machine learning model to predict the inference execution time, which brings a non-negligible model training cost.
To mitigate such a training cost, COSE~\cite{akhtar2020cose} and MBS~\cite{ali2022optimizing} employ Bayesian Optimization to identify the cost-effective resource provisioning plan, which highly depends on initial sampling and seeding.
In contrast, \Name{} identifies a cost-effective function provisioning plan using an inference performance model.
ElasticFlow~\cite{gu2023elasticflow} proposes a greedy algorithm to allocate GPU resources on serverless platforms dynamically which only works for DNN training workloads. 


\textbf{Performance modeling of DNN inference.} To model DNN inference performance, BARISTA~\cite{bhattacharjee2019barista} employs the maximum likelihood estimation approach to obtain the distribution of inference latency. BATCH~\cite{ali2020batch} leverages multi-variable parametric regression to model the inference latency as a deterministic distribution. Instead of identifying the latency distribution, \Name{} develops an analytical model of the average and maximum inference latency, which significantly reduces workload profiling overhead.
A recent work named iGniter~\cite{xu2022igniter} emphasizes performance interference for GPU-based inferences. INFless~\cite{yang2022infless} adopts a lightweight Combined Operator Profiling method to predict inference latency with GPU spatial-sharing. 
In contrast, we focus on modeling inference latency on public heterogeneous functions by explicitly considering the GPU time-slicing scheduling mechanism. 


\section{Conclusion and Future Work}
\label{sec:conclusion}

This paper presents the design and implementation of \Name{}, a cost-efficient resource provisioning framework that achieves predictable performance for multi-SLO DNN inference with heterogeneous serverless functions. 
\Name{} consists of a lightweight performance and cost model of DNN inference on heterogeneous functions and a two-stage merging strategy,
which judiciously batches the multi-SLO DNN inference requests into application groups and provisions each group with adequate CPU or GPU function resources. Prototype experiments on Alibaba Cloud Function Compute demonstrate that \Name{} can deliver predictable DNN inference performance on serverless platforms while saving the monetary cost by up to \CostReduce{} compared to the state-of-the-art methods, yet with acceptable runtime overhead.

We plan to extend \Name{} in two directions: (1) supporting large model inference by leveraging model partitioning, and (2) supporting other public serverless platforms (\emph{e.g.,} AWS Lambda) when GPU functions are enabled.


\bibliographystyle{IEEEtran}
\bibliography{ref}

\appendix

\subsection{Derivation of Equivalent Batching Timeout}
\label{sec:appendix-eq}

The probability that the first request arriving in the buffer belongs to App$_i$ ($i=1$ or $2$) is $\eta_i = \frac{r_i}{r_1 + r_2}$, where $r_1$ and $r_2$ represent the arrival rates for App$_1$ and App$_2$, respectively. If there are no subsequent requests after the first one, the request waiting time becomes $T_i$. However, if another request (belongs to App$_j$, $j=1$ or $2$) arrives in the buffer before $T_i$ elapses, the request waiting time is updated to $\min(T_i, t+T_j)$, where $t$ is the time at which the App$_j$ request arrives.

More specifically, we assume that the buffer capacity limit is infinite and $T_1$ is shorter than $T_2$. We consider the following scenario: if the first request in the buffer is from App$_2$ and a request from App$_1$ arrives within a time frame of $T_2-T_1$, the request waiting time is updated to $t+T_1$, where $t$ represents the time at which the App$_1$ request arrives. 
However, the request waiting time remains to be $T_2$, if no request from App$_1$ arrives within $T_2-T_1$.

Based on the above, App$_1$ follows a Poisson distribution with the request rate $r_1$. The equivalent batching timeout of the application group $\mathcal{X}$ with App$_1$ and App$_2$ can be calculated by
\begin{equation}
\begin{split}
    T^{\mathcal{X}}
    = & \eta_1 \cdot T_1 + \eta_2 \cdot \big(\int_{0}^{T_2-T_1}f(t) \cdot (t+T_1) dt +\\
    & T_2 \cdot (1 - \int_{0}^{T_2-T_1}f(t)dt)\big)\\
    = & T_1 + \eta_2 \cdot \frac{1-\exp(-r_1\cdot(T_2-T_1))}{r_1},
\end{split}
\end{equation}
where $f(t) = r_1 \cdot \exp(-r_1 \cdot t)$ represents the probability density function indicating the likelihood of a request from App$_1$ arriving at time $t$.

\subsection{Proof of Theorem 1}
\label{sec:appendix-theorem1}

\begin{proof}
By substituting Eq.~\eqref{eq-cpu-lat} -- \eqref{eq-t-exp} into Eq.~\eqref{eq-cost}, the average monetary cost of a group $\mathcal{X}$ configured with CPU functions can be reduced to
\begin{equation}
    C^{\mathcal{X}} = \frac{1}{b^{\mathcal{X}}} \cdot [(\alpha_{b^{\mathcal{X}}}^{avg}\cdot \exp\big(-\frac{c}{\beta_{b^{\mathcal{X}}}^{avg}}\big) + \gamma_{b^{\mathcal{X}}}^{avg}) \cdot c \cdot K_1 + K_3].
\end{equation}
Based on the equation above, $C^{\mathcal{X}}$ has at most one relative minimum point, which can obtained by solving the equation of $(C^{\mathcal{X}})^{'} = 0$ using the binary search method.
The derivative of $C^{\mathcal{X}}$ with respect to $c$ is given by
\begin{equation}
    (C^{\mathcal{X}})^{'} = \frac{K_1}{b^{\mathcal{X}}} \cdot [\alpha_{b^{\mathcal{X}}}^{avg} \cdot (1-\frac{c}{\beta_{b^{\mathcal{X}}}^{avg}}) \cdot \exp\big(-\frac{c}{\beta_{b^{\mathcal{X}}}^{avg}}\big) + \gamma_{b^{\mathcal{X}}}^{avg}].
\end{equation}
Based on the above, the minimum value of $C^{\mathcal{X}}$ can only be obtained at the relative minimum point $c_0$ or boundary points $c_{min}$ and $c_{max}$.
Accordingly, the minimal monetary cost can be achieved with the allocated vCPU cores set as
\begin{equation}
    c^{*} = \arg \min C^{\mathcal{X}}(c), \quad c \in \{c_0, c_{min}, c_{max}\}.
\end{equation}

\end{proof}

\subsection{Proof of Theorem 2}
\label{sec:appendix-theorem2}

\begin{proof}
By substituting Eq.~\eqref{eq-cpu-lat} -- \eqref{eq-t-exp} into Eq.~\eqref{eq-cost}, the average monetary cost of a group $\mathcal{X}$ configured with GPU functions can be obtained by
\begin{equation}
    C^{\mathcal{X}} = M_{max} \cdot K_2 \cdot \xi_1 + \frac{M_{max} \cdot \xi_2 + K_3}{b^{\mathcal{X}}}.
\end{equation}
According to the equation above, the cost $C^{\mathcal{X}}$ can be minimized, as $b^{\mathcal{X}}$ reaches its maximum value.
We set $t^w$ as the maximum value (\emph{i.e.,} $s^{w} - L_{max}^{t}$) that satisfies the Constraint~\eqref{eq-cons-SLO}. As $b^{\mathcal{X}}$ increases, the corresponding value of $t^w$ for applications in a group $\mathcal{X}$ decreases, leading to a reduced equivalent batching timeout $T^{\mathcal{X}}$, which is likely to violate Constraint~\eqref{eq-cons-rtb}.
Accordingly, the maximum value of $b^{\mathcal{X}}$ is achieved when $\lfloor r^{\mathcal{X}} \cdot T^{\mathcal{X}} \rfloor + 1 = b^{\mathcal{X}}$, which can be determined efficiently through the binary search method.
\end{proof}

\end{document}